\documentclass[11pt]{article}

\usepackage[centertags]{amsmath}
\usepackage{amsfonts}
\usepackage{amssymb}
\usepackage{amsthm}
\usepackage{newlfont}
\usepackage{epsfig}
\usepackage{amscd}

\newcommand{\RR}{{\mathbb R}}

\newcommand{\CC}{{\mathbb C}}
\newcommand{\beq}{\begin{equation}}
\newcommand{\eeq}{\end{equation}}
\newcommand{\ba}{\begin{array}}
\newcommand{\ea}{\end{array}}
\newcommand{\bea}{\begin{eqnarray}}
\newcommand{\eea}{\end{eqnarray}}

\usepackage{graphicx}
\usepackage{color}
\usepackage{transparent} 

\newtheorem{theorem}{Theorem}
\newtheorem{lemma}{Lemma}[section]
\newtheorem{proposition}{Proposition}
\newtheorem{definition}{Definition}

\newtheorem{remark}{Remark}

\numberwithin{equation}{section}
\allowdisplaybreaks

\begin{document}

\begin{center}
{\large   \bf The Cauchy problem for the Pavlov equation with large data} 
 
\vskip 15pt

{\large  Derchyi Wu }

\vskip 8pt

{\texttt\tiny{
 Institute of Mathematics, Academia Sinica, 
Taipei, Taiwan}}

\vskip 5pt

{\today}

\end{center}

\begin{abstract}
The Pavlov equation is one of the simplest integrable systems of vector fields arising from  various problems of mathematical physics and differential geometry which are intensively studied in recent literature. In this report,  solving a nonlinear Riemann-Hilbert problem via a Newtonian iteration scheme, we complete the inverse scattering theory and prove a short time unique solvability of the Cauchy problem of  the Pavlov equation   with large initial data. 
\end{abstract}

\vskip10pt 
\noindent

\section{Introduction}\label{S:introduction}
Integrable dPDEs (dispersionless partial differential equations), including the dispersionless Kadomtsev-Petviashvili equation \cite{Lin}, \cite{Timman}, \cite{ZK}, the first and second heavenly equations of Plebanski \cite{Plebanski}, the dispersionless 2D Toda (or Boyer-Finley) equation \cite{BF}, \cite{FP}, and the Pavlov equation \cite{D02}, \cite{Ferapontov}, \cite{Pavlov}, are defined by a commutation $[L ,M]=0$ of pairs of one-parameter families of vector fields.  They arise in various problems of mathematical physics and are intensively studied  recently.

Due to lack of dispersion, integrable dPDEs may or may not exhibit a gradient catastrophe at finite time.  Since the Lax operators are vector fields, the  Fourier transform theory used in soliton theory for proving the existence of  eigenfunctions fails and the inverse problem is intrinsically 
nonlinear for integrable dPDEs, unlike the $\bar\partial$-problem formulated for general  soliton equations \cite{ABF}, \cite{BC81}.  At last, no explicit regular localized solutions, like solitons or lumps, exist for integrable dPDEs. 
Therefore, it is important to solve the inverse scattering problem for integrable dPDEs.  A formal inverse scattering theory has been recently constructed, including i) to solve their Cauchy problem, ii) obtain the longtime behavior of solutions, iii) costruct distinguished classes of exact implicit solutions, iv) establish if, due to the lack of dispersion, the nonlinearity of the PDE is ``strong enough'' to cause the gradient catastrophe of localized multidimensional disturbances, and v) to study analytically the breaking mechanism \cite{MS0}-\cite{MS10}.

The Pavlov equation,
\begin{equation}\label{E:pavlov-LD}
\begin{array}{c}
v_{xt}+v_{yy}+v_xv_{xy}-v_yv_{xx}=0,\\
v=v(x,y,t)\in\mathbb R,~~x,y,t\in\mathbb R, 
\end{array}
\end{equation}arising in the study of integrable hydrodynamic chains  \cite{Pavlov}, and in differential geometry as a particular example of Einstein - Weyl metric \cite{D02}, is the simplest integrable dPDE available in the literature \cite{Pavlov}, \cite{Ferapontov}, \cite{D02}. 
 It was first derived in \cite{Duna1} as a conformal symmetry of the second heavenly equation. In our previous work \cite{GSW}, 
we solve the forward problem  via a  Beltrami-type equation, a first order PDE,  and a shifted Riemann-Hilbert problem and the inverse problem  by a nonlinear integral equation under a small data constraint. More precisely, we justify the complex eigenfunction $\Phi(x,z,\lambda)$, defined by
\beq\label{E:pavlov-ce}
\begin{array}{rl}
\partial_y\Phi(x,y,\lambda)+\left(\lambda+v_x\right)\partial_x\Phi(x,y,\lambda)=0,&\quad x,\,y\in\mathbb R,\,\lambda\in\mathbb C^\pm\\
\Phi(x,y,\lambda)-\left(x-\lambda y\right)\ \to\ 0,&\quad |x|,\,|y|\to\infty,
\end{array}
\eeq
is link to the real eigenfunction $\varphi(x,y,\lambda)$, defined by
\begin{equation} \label{E:pavlov-lax}
\begin{array}{rl}
\partial_y\varphi+\left(\lambda+v_x\right)\partial_x\varphi=0,&\quad\textit{ $x,\,y\in\mathbb R$, $\lambda\in\mathbb R$}\\
\varphi-(x-\lambda y )\to 0, & \quad\textit{ $y\to -\infty$}
\end{array}
\end{equation}
by the equation
\bea
\Phi^{-}(x,y,\lambda)
&=&\varphi(x,y,\lambda)+\chi^-(\varphi(x,y,\lambda),\lambda)\label{E:psi-phi+-pavlov}
\eea where $\chi^-(\xi,\lambda)$ satisfies a shifted Riemann-Hilbert problem 
\beq
\label{E:shifted-RH-pavlov}
\begin{split}
\sigma(\xi,\lambda)+\chi^+(\xi+\sigma(\xi,\lambda),\lambda)-\chi^-(\xi,\lambda)=0,&\quad \xi\in \RR,\\
\sigma(\xi,\lambda)=\lim_{y\to\infty}\left(\varphi(\xi,y,\lambda)-\xi\right).
\end{split}
\eeq  
Therefore, the inverse scattering theory reduces to deriving  uniform $\xi$-, $\lambda$-asymptotic estimates of $\chi(\xi,\lambda)$ in the direct problem, and solving the associated Riemann-Hilbert problem via the nonlinear integral equation
\begin{equation}\label{E:nonlinear-inv-1}
\psi(x,y,t,\lambda)+\chi^-_R(\psi(x,y,t,\lambda),\lambda)=x-\lambda y-\lambda^2 t-\frac 1\pi\int_R\frac{\chi^-_I(\psi(x,y,t,\zeta),\zeta)}{\zeta-\lambda}d\zeta
\end{equation} in the inverse problem. 
However, for technical reasons, we have to impose two strong assumptions, the compact support condition and small data constraint on the initial data $v(x,y,0)$, to make the above resolution scheme possible \cite{GSW}.

In this report, we aim to solving a large data problem. The difficulties are to   find an effective scheme to solve the nonlinear inverse problem with large data and to establish a Fredholm alternative for the linearization. To do it, we transform the conjunction formula \eqref{E:psi-phi+-pavlov} to   
\begin{equation}\label{E:nrh-in}
\begin{array}{c}
\Phi^+(x,y,\lambda)=\Phi^-(x,y,\lambda)+R(\Phi^-(x,y,\lambda),\lambda).
 \end{array}
 \end{equation}and the inverse problem \eqref{E:nonlinear-inv-1} to the nonlinear Riemann-Hilbert problem 
\begin{equation}\label{E:nrh-in-normalize}
\begin{array}{cl}
\Psi^+(x,y,t,\lambda)=\Psi^-(x,y,t,\lambda)+R(\Psi^-,\lambda), &\lambda\in\RR,\\
\partial_{\overline\lambda}\Psi(x,y,t,\lambda)=0, &\lambda\in\CC^\pm,\\
\Psi(x,y,t,\lambda)-(x-\lambda y-\lambda^2 t)\to 0, &|\lambda|\to\infty,\\
\Psi(x,y,0,\lambda)=\Phi(x,y,\lambda).&
 \end{array}
 \end{equation}\cite{MS3}-\cite{MS7}. Furthermore, we adopt a Newtonian iteration approach to study  \eqref{E:nrh-in-normalize}. The linearization of the Newtonian process 
 turns out to be a non homogeneous Riemann-Hilbert problem of which a Fredholm theory has been well investigated \cite{Ga66}. 
Consequently,  
an index zero condition on the linearization $1+\partial_\zeta R$, vital for a Fredholm theory, and a deformation property of $R(\Phi^-(x,y,\lambda),\lambda)$, needed in the Newtonian iteration scheme, should be justified to make the Newtonian iteration approach feasible.
 
The dispersion term  $x-\lambda y-\lambda^2 t$  causes the estimates, necessary for the above Newtonian iteration approach, growing  inevitably unbounded when  $|y|\to\infty$ or $t$ gets larger. {Precisely, the obstruction to the global solvability is the non existence of $R(\omega^-+x-\lambda y-\lambda^2t,\lambda)$ when $t$ gets larger.} Hence only  a local solvability of the nonlinear Riemann-Hilbert problem is achieved in general. Moreover, without compact support constraints, the quadratic dispersion term  destroys the $L^1(\mathbb R,d\lambda)$ property of $\partial R/\partial t$ as $t>0$. So only a local solvability of  the Lax pair of the Pavlov equation 
holds. To derive a local solvability of the Cauchy problem of the Pavlov equation, we still need to impose the compact support condition \cite{GSW}.

 
The contents of the paper are as follows. In Section \ref{S:forward-problem} and \ref{S:forward-problem-2}, removing the compact support condition on the initial data $v(x,y)$ and using the $L^p(\RR,d\xi)$ - Hilbert transform theory, we 
derive  
$\lambda$-asymptotics of the real eigenfunction $\varphi(x,y,\lambda)$ and the solution to the shifted Riemann-Hilbert problem $\chi(\xi,\lambda)$ 
via  uniform $L^p(\RR,d\xi)$ estimates on them. In Section \ref{S:NRH-dp}, we derive the conjunction formula \eqref{E:nrh-in} between complex eigenfunctions  $\Phi^\pm(x,y,\lambda)$ and define $R(\zeta,\lambda)$ as the spectral data. An index zero condition on $1+\partial_\zeta R$ and a deformation property of $R(\Phi^-(x,y,\lambda),\lambda)$ are justified as well. In Section \ref{S:NRH}, via a Newtonian iteration method and using the $L^p(\RR,d\lambda)$ - Hilbert transform theory, 
we solve the short time unique existence of the nonlinear Riemann-Hilbert problem \eqref{E:nrh-in-normalize} with prescribed initial data. The short time unique existence of the Lax pair  and of the Cauchy problem of the Pavlov equation \eqref{E:pavlov-LD} are established in Secion \ref{S:lax-cauchy}.

\noindent
{\bf Acknowledgments}. The author would  like to thank  S. Manakov and P. Santini for their pioneer contribution in the IST of dispersionless PDEs and  S. X. Chen of Fudan university for helpful discussion in the Newtonian iteration method. The author feel very grateful to P. Grinevich for many brilliant inputs in discussion. The author was  partially supported by NSC 103-2115-M-001 -003 -2.
\section{The forward problem I: the real eigenfunctions }\label{S:forward-problem}
The forward problem of the Pavlov equation with compactly supported initial data has been solved in \cite{GSW}. Removing the compact support restriction, we prove the existence of real eigenfunction $\varphi(x,y,\lambda)$ and derive uniform estimates of scattering data $\sigma(\xi,\lambda)$  
in this section.

Throughout this report, 
\begin{eqnarray*}
&&\mathfrak S =\{f:\mathbb R^2\to\mathbb R| \textit{ $f(x,y)$ is Schwartz in $x$, $y$}\},\\ 
&&L^p(\mathbb R, d\lambda) =\{f:\mathbb R\to\mathbb C|\ |f|_{L^p(\mathbb R, d\lambda)}= (\int_{\mathbb R}| f(\lambda)|^pd\lambda)^\frac 1p<\infty\},\\
&&H^p(\mathbb R, d\lambda) =\{f:\mathbb R\to\mathbb C|\ f,\ \partial_\lambda f\in L^p(\mathbb R, d\lambda)\}.
\end{eqnarray*}
Consider, for each fixed $\lambda\in\RR$,  
\begin{eqnarray}
\partial_y\varphi_\pm+(\lambda+ v_x)\partial_x\varphi_\pm=0,&&\textit{for $x$, $y\in \mathbb R$,}\label{E:pavlov-lax-r}\\
\varphi_\pm(x,y,\lambda)-\xi\to 0, &&\textsl{as $y\to \pm\infty$,}\label{E:pavlov-lax-bdry-r}
\end{eqnarray}
where
$\xi=x-\lambda y$. The solvability and uniqueness of the boundary value problem of the first order partial differential equation (\ref{E:pavlov-lax-r}), (\ref{E:pavlov-lax-bdry-r})  is shown by solving the ordinary differential equation \cite{GSW}
\beq
\frac {d x}{dy}=\lambda+v_x(x,y), \ \ x=x(y;x_0,y_0,\lambda),  \ \ x(y_0;x_0,y_0,\lambda)=x_0, 
\eeq
or, equivalently, 
\beq
\label{eq:def-h}
\begin{array}{c}
\frac {d h}{dy}=v_x(h+\lambda y,y), \\
 h=h(y;\xi_0,y_0,\lambda)= x(y;x_0,y_0,\lambda) -\lambda y, \\ 
 h(y_0;\xi_0,y_0,\lambda)=x_0-\lambda y_0=\xi_0.
 \end{array}
\eeq
Hence
\bea
h(y';x-\lambda y,y,\lambda)\label{E:h}
& =&\xi +\int_{y}^{y'} v_x\big(h(y'';x-\lambda y,y,\lambda)+ \lambda y'',y''\big)dy'',\\
\varphi_\pm(x,y,\lambda)\label{E:phi}
&= &h(\pm\infty;x-\lambda y,y,\lambda)\\
&=&\xi +\int_{y}^{\pm\infty} v_x\big(h(y'';x-\lambda y,y,\lambda)+ \lambda y'',y''\big)dy'',\nonumber
\eea and 
\begin{gather}
|\partial_{y'}^\beta\partial_\lambda^\alpha\partial_{\xi}^{\mu} \left(h(y';\xi,y,\lambda)-\xi\right)|\le C_{\mu,\alpha,\beta}.\label{E:h-mu-x}
\end{gather}

The \textbf{\emph{real eigenfunction}} of the Pavlov equation is defined by 
\beq\label{E:real-eigenfunction-definition}
\varphi(x,y,\lambda)=\varphi_-(x,y,\lambda)
\eeq and the {\textit{\textbf{scattering data}} $\sigma(\xi,\lambda)$ is defined as
\beq\label{E:definition-sigma}
\varphi_+(x,y,\lambda)=\varphi_-(x,y,\lambda)+\sigma(\varphi_-(x,y,\lambda),\lambda),
\eeq where
\begin{equation}\label{E:sigma}
\begin{split}
\sigma(\xi,\lambda)=&\,h(\infty;\xi,-\infty,\lambda)-\xi\\
=&\,\int_{-\infty}^{\infty} v_x\big(h(y'';\xi,-\infty,\lambda)+ \lambda y'',y''\big)dy''
\end{split}
\end{equation}

\begin{lemma}\label{L:sigma}
Suppose $v\in\mathfrak S$. Then
\begin{gather}
0<C_1<1+\partial_\xi\sigma(\xi,\lambda)<C_2.\label{E:positive}
\end{gather}
\end{lemma}
\begin{proof} 
Comparing \eqref{E:h} and \eqref{E:sigma}, to prove \eqref{E:positive}, it suffices to show that 
\begin{equation}\label{E:posi}
0<c_1\le \partial_xh(y';x-\lambda y, y,\lambda)\le c_2
\end{equation}
for  two constants $c_1$ and $c_2$. Taking derivatives on \eqref{eq:def-h}, one obtains
\begin{gather}
\frac \partial{\partial y'}\frac {\partial h(y';x-\lambda y,y,\lambda)}{\partial x}=v_{xx}(h+\lambda y', y')\frac {\partial h(y';x-\lambda y,y,\lambda)}{\partial x},\label{E:gronwall}\\
\frac {\partial h(y;x-\lambda y,y,\lambda)}{\partial x}=1.\label{E:gronwall-b}
\end{gather}
Therefore, 
\[
\begin{array}{c}
-\int_{\mathbb R} |u_{xx}(h+\lambda y'', y'')|dy''\le \log \frac {\partial h(y';x-\lambda y,y,\lambda)}{\partial x} \le \int_{\mathbb R} |u_{xx}(h+\lambda y'', y'')|dy''
\end{array}
\]and 
\begin{equation}\label{E:positive-h}
e^{-\int_{\mathbb R} |u_{xx}(h+\lambda y'', y'')|dy''}\le  \frac {\partial h(y';x-\lambda y,y,\lambda)}{\partial x} \le e^{\int_{\mathbb R} |u_{xx}(h+\lambda y'', y'')|dy''}
\end{equation}
So  
\beq
0<C_1<\partial_\xi h(y',\xi,y,\lambda)<C_2,\label{E:positive-h-0}
\eeq
 by the Schwartz condition and \eqref{E:posi} is proved.
\end{proof}

\begin{proposition}\label{P:direct-sigma-asy-0}
Suppose $v\in\mathfrak S$ and $p\ge 1$. Then the scattering data $\sigma$  satisfies 
\begin{gather}
|{\partial_\lambda^\nu}\partial_\xi^{\mu}\sigma(\xi,\lambda)|_{L^\infty}\le\frac{C}{1+|\lambda |^{2+\mu+\nu}},\label{E:l1-condition-lambda-tau-sigma-0-new-infty} \\
|{\partial_\lambda^\nu}\partial_\xi^{\mu}\sigma(\xi,\lambda)|_{L^p(d\xi)}\le\frac{C}{1+|\lambda |^{2+\mu+\nu-\frac 1p}}.\label{E:l1-condition-lambda-tau-sigma-0-new} 
\end{gather} 
\end{proposition}
\begin{proof} The method in the proof of Proposition 3.2 in \cite{GSW} can be applied to the non-compact case as well. Indeed, 
let
\beq\label{E:change}
x= h(y';\xi,-\infty,\lambda)+\lambda y'=\xi+\lambda y'+\int^{y'}_{-\infty} v_{x}dy''.
\eeq 
For  $\lambda\gg 1$, from the implicit function theorem, there exist $y'=H(\xi,x,\lambda)$ and $H_1(\xi,x,\lambda)$ so that
\beq
H(\xi,x,\lambda)=\frac {-\xi+x}\lambda+\frac {H_1(\xi,x,\lambda)}{\lambda^2}.\label{E:exp}
\eeq 
So
\beq\label{E:sigma-2}
\sigma(\xi,\lambda)=- \frac{H_1(\xi,\infty,\lambda)}\lambda.
\eeq
From \eqref{E:change}, one has
\beq\label{E:sig-x}
\frac{\partial y'}{\partial x}=\frac 1{\lambda+v_x(x,H(\xi,x,\lambda))},
\eeq or, equivalently
\begin{gather}
\frac{\partial H_1(\xi,x,\lambda)}{\partial x}=-\frac{v_x(x,\frac{-\xi+x}\lambda+\frac {H_1}{\lambda^2})}{1+{v_x(x,\frac{-\xi+x}\lambda+\frac {H_1}{\lambda^2})}/\lambda},\label{E:sigma-h-1}\\
H_1(\xi,-\infty,\lambda)=0.\label{E:sigma-h-0}
\end{gather}
Define
\begin{gather*}
\hat \lambda=\frac 1\lambda,\quad\hat\xi=\frac\xi\lambda,\quad
\hat H_1(\hat\xi,x,\hat\lambda)=H_1(\frac {\hat\xi}{\hat\lambda},x,\frac 1{\hat\lambda}).
\end{gather*}
So \eqref{E:sigma-h-1} and \eqref{E:sigma-h-0} turn into
\begin{gather}
\frac{\partial \hat H_1(\hat\xi,x,\hat\lambda)}{\partial x}=-\frac{v_x(x,-\hat\xi+\hat \lambda x+\hat\lambda^2\hat H_1(\hat\xi,x,\hat\lambda))}{1+\hat\lambda v_x(x,-\hat\xi+\hat \lambda x+\hat\lambda^2\hat H_1(\hat\xi,x,\hat\lambda))},\label{E:sigma-h-1-hat}\\
\hat H_1(\hat\xi,-\infty,\hat\lambda)=0.\label{E:sigma-h-0-hat}
\end{gather}
For $|\hat\lambda|<\frac 1{2\max|v_x(x,y)|}$, the right hand side of \eqref{E:sigma-h-1-hat} is smooth in $\hat\xi$, $\hat\lambda$. Expanding it at $\hat\lambda=0$ and using the boundary condition \eqref{E:sigma-h-0-hat}, we obtain
\beq\label{E:sigma-h-1-hat-0}
\begin{split}
&\hat H_1(\hat\xi,\infty,\hat\lambda)\\
=\,&-\int_{-\infty}^\infty\frac{v_x(x,-\hat\xi+\hat \lambda x+\hat\lambda^2\hat H_1(\hat\xi,x,\hat\lambda))}{1+\hat\lambda v_x(x,-\hat\xi+\hat \lambda x+\hat\lambda^2\hat H_1(\hat\xi,x,\hat\lambda))}dx\\
=\,&-\int_{-\infty}^\infty v_x(x,-\hat\xi)dx+\mathcal O(\hat\lambda)\\
=\,&\mathcal O(\hat\lambda)
\end{split}
\eeq by the mean value theorem, fundamental theorem of calculus, and $v\in\mathfrak{G}$. From the Hadamard's lemma it follows
\[
\frac{\sigma(\hat\xi/\hat\lambda,1/\hat\lambda)}{\hat\lambda^2}=-\frac{\hat H_1(\hat\xi,\infty,\hat\lambda)}{\hat\lambda}
\]
is regular in $\hat\xi$, $\hat\lambda$ for $|\hat\lambda|<\frac 1{2\max|v_x(x,y)|}$. Consequently, by
\beq\label{E:change-o-v}
\partial_\lambda=-\hat\lambda^2\partial_{\hat\lambda}-\hat\lambda\hat\xi\partial_{\hat\xi},\quad\partial_\xi=\hat\partial_{\hat\xi},
\eeq
we obtain
\begin{gather}
|{\partial_\lambda^\nu}\partial_\xi^{\mu}\sigma(\xi,\lambda)|_{L^\infty}\le\frac{C}{1+|\lambda |^{2+\mu+\nu}},\label{E:l1-condition-lambda-tau-sigma-0-new-prep} 
\end{gather}
if $\lambda\gg 1$. By analogy, one can show \eqref{E:l1-condition-lambda-tau-sigma-0-new-prep} if $\lambda\ll -1$. Thus \eqref{E:l1-condition-lambda-tau-sigma-0-new-infty} is proved. Moreover, \eqref{E:l1-condition-lambda-tau-sigma-0-new} follows from \eqref{E:sigma-h-1-hat-0}, \eqref{E:change-o-v},  and the Minkowski inequality.

\end{proof}

\section{The forward problem II: the complex eigenfunctions }\label{S:forward-problem-2}

With the positivity property (\ref{E:positive}), one has the unique solvability of the shifted Riemann-Hilbert problem \cite{Ga66}
\beq
\label{E:shifted-RH}
\begin{split}
\sigma(\xi,\lambda)+\chi^+(\xi+\sigma(\xi,\lambda),\lambda)-\chi^-(\xi,\lambda)=0,&\quad \xi\in \RR,\\
\hskip1in\partial_{\bar\xi}\chi=0,&\quad \xi\in\CC^\pm,\\
\chi\to 0,\ |\xi|\to\infty,&\quad \xi\in\CC.
\end{split}
\eeq 

Applying \eqref{E:l1-condition-lambda-tau-sigma-0-new-infty}, \eqref{E:l1-condition-lambda-tau-sigma-0-new}, and the boundedness of the Hilbert transform, we show that
\begin{proposition}\label{P:shifted-RH-0-xi-lambda} 
If $v\in\mathfrak S$ and $p> 1$, then the shifted Riemann-Hilbert problem (\ref{E:shifted-RH}) admits a unique bounded solution $\chi$ satisfying
\begin{eqnarray}
|{\partial_\lambda^\nu}\partial_\xi^{\mu}\chi^-(\xi,\lambda)|_{ L^\infty}&\le&\frac {C}{1+|\lambda|^{2+\mu+\nu-\frac 1p}},\quad \forall \xi\in\RR,\ \ \forall \lambda\in \RR,\label{E:nonlinear-scattering-data-asym-0-revised-1-plus}\\
|{\partial_\lambda^\nu}\partial_\xi^{\mu}\chi(\xi,\lambda)|_{ L^\infty}&\le&\frac {C}{1+|\lambda|^{2+\mu+\nu-\frac 1p}},\quad \forall\xi\in\CC^\pm,\ \ \forall\lambda\in \RR.\label{E:nonlinear-scattering-data-asym-0-revised-2-plus}
\end{eqnarray} 
\end{proposition}
\begin{proof} In \cite{MS9}, the unique solvability of the shifted Riemann-Hilbert problem (\ref{E:shifted-RH}) has been justified by converting it to
the following linear equation  
\beq
\label{E:shifted-RH-fredholm}
\chi^-(\xi,\lambda)-\frac 1{2\pi i}\int_\RR f(\xi,\xi',\lambda)\chi^-(\xi',\lambda)d\xi'+ g(\xi,\lambda)=0,
\eeq
where
\beq \label{E:shifted-RH-fredholm-c-2}
\begin{split}
&f(\xi,\xi',\lambda)=  \frac {\partial_{\xi'}s(\xi',\lambda)}{s(\xi',\lambda)-s(\xi,\lambda)}-\frac 1{\xi'-\xi},\\
&g(\xi,\lambda)=  -\frac 12\sigma(\xi,\lambda)+\frac 1{2\pi i}\int_\RR \frac {\partial_{\xi'}s(\xi',\lambda)}{s(\xi',\lambda)-s(\xi,\lambda)}\sigma(\xi',\lambda)d\xi',\\
&s(\xi,\lambda)=\, \xi+\sigma(\xi,\lambda).
\end{split}
\eeq Hence it yields to showing the uniform estimate \eqref{E:nonlinear-scattering-data-asym-0-revised-1-plus} without the compactly supported condition.

Taking derivatives of (\ref{E:shifted-RH-fredholm})  and applying Proposition \ref{P:direct-sigma-asy-0},  we obtain
\begin{eqnarray}
&&{\partial_\lambda^\nu}\partial_\xi^{\mu}\chi^-(\xi,\lambda)\label{E:general-k}\\
&=&\frac 1{2\pi i}\int_\RR \frac {{\partial_\lambda^\nu}\left(\partial_\xi+\partial_{\xi'}\right)^{\mu}\left[k(\xi,\xi',\lambda)\chi^-(\xi',\lambda)\right]}{\xi'-\xi}d\xi'\nonumber\\
&-&\frac 1{2\pi i}\int_\RR \frac {{\partial_\lambda^\nu}\left(\partial_\xi+\partial_{\xi'}\right)^{\mu}\left[k(\xi,\xi',\lambda)\sigma(\xi',\lambda)\right]}{\xi'-\xi}d\xi'\nonumber\\
&+&\frac 1{2\pi i}\int_\RR \frac{{\partial_\lambda^\nu}\partial_{\xi'}^{\mu}\sigma(\xi',\lambda)}{\xi'-\xi}d\xi'+\frac 12{\partial_\lambda^\nu}\partial_\xi^{\mu}\sigma\nonumber
\end{eqnarray}where
\begin{equation}\label{E:k-definition}
k(\xi,\xi',\lambda)=\frac{s(\xi,\lambda)-s(\xi',\lambda)-\partial_{\xi'}s(\xi',\lambda)(\xi-\xi')}{s(\xi,\lambda)-s(\xi',\lambda)}.
\end{equation}Applying Proposition \ref{P:direct-sigma-asy-0}, Hadamard's lemma, and an induction argument, one can derive
{\begin{equation}\label{E:k-hardmard}
\begin{split}
|\partial_\lambda^\nu\partial_\xi^{\mu_1}\partial_{\xi'}^{\mu_2} k(\xi,\xi',\lambda)|_{L^\infty}\le \frac C{1+|\lambda|^{\mu_1+\mu_2+\nu}},
\end{split}
\end{equation}
and then}
\begin{equation}\label{E:hadamard-chi-sigma}
\begin{split}
|{\partial_\lambda^\nu}\left(\partial_\xi+\partial_{\xi'}\right)^{\mu}k\chi^--k{\partial_\lambda^\nu}\partial_{\xi'}^{\mu}\chi^-|_{L^p(\mathbb R,d\xi)}&\le \frac C{1+|\lambda|^{2+\mu+\nu-\frac 1p}},\\
|{\partial_\lambda^\nu}\left(\partial_\xi+\partial_{\xi'}\right)^{\mu}k\sigma
|_{L^p(\mathbb R,d\xi)}&\le \frac C{1+|\lambda|^{2+\mu+\nu-\frac 1p}}.
\end{split}
\end{equation}Plugging (\ref{E:hadamard-chi-sigma}) into (\ref{E:general-k}), we obtain
\[
|{\partial_\lambda^\nu}\partial_\xi^{\mu}\chi^--\mathcal K\left({\partial_\lambda^\nu}\partial_\xi^{\mu}\chi^-\right)|_{L^p(\mathbb R,d\xi)}\le \frac C{1+|\lambda|^{2+\mu+\nu-\frac 1p}},
\]where
\[
\mathcal K\psi(\xi,\lambda)=\frac 1{2\pi i}\int_{\RR}f(\xi,\xi',\lambda)\psi(\xi',\lambda)d\xi'.
\]
Hence the boundedness of the Hilbert transform on $L^p$ for $p>1$, 
\bea
|{\partial_\lambda^\nu}\partial_\xi^{\mu}\chi^-(\xi,\lambda)|_{L^p(\RR,d\xi)}&\le &\frac C{1+|\lambda|^{2+\mu+\nu-\frac 1p}},\quad|\lambda|\gg 1.\label{E:chi-xi-xi-2-V-plus}
\eea So
\bea
|{\partial_\lambda^\nu}\partial_\xi^{\mu}\chi^-(\xi,\lambda)|_{ L^\infty}&\le&\frac {C}{1+|\lambda|^{2+\mu+\nu-\frac 1p}},\quad |\lambda|\gg 1,
\eea from the Sobolev imbedding theorem. Finally \eqref{E:nonlinear-scattering-data-asym-0-revised-1-plus} is achieved for $\forall \lambda\in\RR$ by continuity. 
\end{proof}

\begin{theorem}\label{T:existence-cpx-pavlov}
If $v\in\mathfrak S$, then for fixed $x,\,y\in\RR$, there 
exists a unique continuous  $\Phi(x,y,\lambda)$ such that $\Phi(x,y,\lambda)$ is holomorphic for $\lambda\in \mathbb C^\pm$ and has continuous limits on both sides of $\lambda\in \RR$  satisfying
\bea
\Phi^{-}(x,y,\lambda)
&=&\varphi(x,y,\lambda)+\chi^-(\varphi(x,y,\lambda),\lambda),\label{E:psi-phi+-reduction}\\
\Phi^+(x,y,\lambda)&=&\overline{\Phi^-(x,y,\lambda)}.\label{E:psi-phi-reduction}
\eea
Moreover,
 $\Phi(x,y,\lambda)$ is a {\textit{\textbf{complex eigenfunction}}}, i.e.  for $\forall\lambda\in\CC^\pm$ fixed,
\begin{gather}
\left(\partial_y+(\lambda+v_x)\right)\partial_{x}\Phi=0,\quad\textit{for $\forall x,\,y\in\mathbb R$, }\label{E:lax-complex-reduction}\\
\Phi(x,y,\lambda)-(x-\lambda y)\to 0,\quad\textit{as $x,\,y\to\infty$,}\nonumber
\end{gather}
and
\begin{gather}
\Phi(x,y,\lambda)={\overline {\Phi(x,y,\bar\lambda)}}, \label{eq:L1_1-sym-reduction}\\
|\partial_\lambda^\nu\partial_x^k\partial_y^h\left(\Phi^\pm-x+\lambda y\right)|_{L^\infty}\le \frac {C}{1+|\lambda|^{2+\nu+k-\frac 1p}}. \label{eq:L1_1-sym-estimate}
\end{gather}

\end{theorem} 
\begin{proof}
We refer to Theorem 3.1 in \cite{GSW} for a detailed proof of this theorem except for \eqref{E:psi-phi+-reduction}. Note that the whole proof there does not require any small data assumption nor a compact support condition. 

Moreover, the conjunction formula \eqref{E:psi-phi+-reduction} for complex and real eigenfunctions has been justified by changes of variables in the proof of Theorem 3.2 in \cite{GSW} under a compact support restriction. However, the same proof can be refined and holds for $v(x,y)\in\mathfrak G$ as well. Indeed, denote $\lambda=\lambda_R+i\lambda_I$, the first change of variables $\mathcal F_1:(x,y)\to (x_1,y_1)$ is defined by
\beq
\label{eq:L1_24}
\left\{
\begin{array}{ll}
x_1=\lim\limits_{y'\rightarrow-\infty} h(y'; x-\lambda_R y,y,\lambda_R)
=\varphi_-(x,y,\lambda_R),&  \ \ y <0\\
x_1=\lim\limits_{y'\rightarrow+\infty} h(y';x-\lambda_R y,y,\lambda_R)
=\varphi_+(x,y,\lambda_R),& 
\ \  y >0\\
y_1 = y & 
\end{array}
\right.
\eeq
 In the new variables,
\beq
\label{eq:L1_25}
L = \partial_y+(\lambda+v_x)\partial_x=\partial_{y_1}+ i\lambda_I\kappa(x_1,y_1) \partial_{x_1} ,
\eeq
where
\beq\label{E:kappa}
\kappa(x_1,y_1)=\frac{\partial\varphi_\pm}{\partial x}(x,y,\lambda_R)|_{(x,y)=\mathcal F_1^{-1}(x_1,y_1)},\ \ y\ne 0.
\eeq By \eqref{E:positive-h}, there exists a pair of positive constants $ C_1$,  
$ C_2$ such that:
\beq
\label{eq:L1_26}
0< C_1 \le\kappa(x_1,y_1) \le  C_2. 
\eeq

Secondly, taking $\lambda_I<0$ from now on without loss of generality, we introduce
$\mathcal F_2:(x_1,y_1)\rightarrow(x_2,y_2)$
\beq
\label{eq:L1_27}
\left\{
\begin{array}{l}
x_2=x_1.\\
y_2 = \lambda_I y_1 ,\\
z_2=x_2-iy_2
\end{array}
\right.
\eeq
In the new variables, \eqref{eq:L1_25} turns into 
\beq
\label{eq:L1_28}
\begin{split}
L =&\lambda_I(\partial_{y_2}+i  \kappa\left(x_2,\frac{y_2}{\lambda_I}\right) \partial_{x_2})\\
=&\lambda_I\left[1+\kappa(x_2,\frac{y_2}{\lambda_I})\right](\partial_{\overline z_2}+r(z_3,\bar z_3, \lambda) \partial_{z_2}).
\end{split}
\eeq
where
\beq
\label{eq:L1_33-new}
r(z_2,\bar z_2,\lambda)=\frac {-1+\kappa(x_2,\frac{y_2}{\lambda_I})}{1+\kappa(x_2,\frac{y_2}{\lambda_I})}|_{(x_2,y_2)=\mathcal F_3^{-1}(z_3,\bar z_3)}<1
\eeq
by \eqref{E:positive-h}.

The last change of variables is 
$\mathcal F_3:(x_2,y_2)\rightarrow z_3$, defined  by
\beq
\label{eq:L1_29}
z_3=x_2-iy_2+\chi(x_2-iy_2,\lambda_R),
\eeq
where $\chi(\xi,\lambda_R)$ is the solution of the shifted Riemann-Hilbert 
problem (\ref{E:shifted-RH}). 
Consequently, \eqref{eq:L1_28} or \eqref{E:lax-complex-reduction} takes the form 
\beq
\label{eq:L1_32}
[\partial_{\bar z_3} + q(z_3,\bar z_3, \lambda) \partial_{z_3}] \Phi = 0
\eeq
where
\beq
\label{eq:L1_33}
\begin{split}
&|q(z_3,\bar z_3,\lambda)|\\
=&|\frac{1+ {\partial\overline\chi}/{\partial\overline z_2}}{1+ {\partial\chi}/{\partial  z_2}}|\,|\frac {-1+\kappa(x_2,\frac{y_2}{\lambda_I})}{1+\kappa(x_2,\frac{y_2}{\lambda_I})}|_{(x_2,y_2)=\mathcal F_3^{-1}(z_3,\bar z_3)}\\
=&|\frac {-1+\kappa(x_2,\frac{y_2}{\lambda_I})}{1+\kappa(x_2,\frac{y_2}{\lambda_I})}|_{(x_2,y_2)=\mathcal F_3^{-1}(z_3,\bar z_3)}\\
<&1
\end{split}
\eeq
by solvability of \eqref{E:shifted-RH} and \eqref{eq:L1_33-new}. It is then natural to consider Beltrami equation (\ref{eq:L1_32}) in the space 
$L^{2+\epsilon}(dz_3d\bar z_3)\cap L^{2-\epsilon}(dz_3d\bar z_3)$ where $\epsilon$ is 
sufficiently small. 
We can write \cite{V62}
\beq
\label{eq:L1_34}
\Phi(z_3,\bar z_3,\lambda) = z_3 + 
\partial^{-1}_{\bar z_3} \alpha(z_3,\bar z_3,\lambda)
\eeq
where
\beq
\label{eq:L1_35}
[1+q(z_3,\bar z_3,\lambda)\partial_{z_3}\partial^{-1}_{\bar z_3}]
\alpha(z_3,\bar z_3,\lambda)+q(z_3,\bar z_3,\lambda)=0.
\eeq

Taking into account \eqref{E:positive-h}, \eqref{E:kappa}, \eqref{eq:L1_29}, (\ref{eq:L1_33}), $v\in\mathfrak G$, Proposition \ref{P:shifted-RH-0-xi-lambda}, and the Fubini theorem, we see that 
\bea
&&|q(z_3,\bar z_3,\lambda) |_{L^{p}}^p \nonumber\\
&<&|{-1+\kappa(x_2,\frac{y_2}{\lambda_I})}|_{(x_2,y_2)=\mathcal F_3^{-1}(z_3,\bar z_3)}|_{L^{p}}^p \nonumber\\
&<&\iint dz_{3,I}dz_{3,R}\nonumber\\
&&|\int_{\frac {y_2}{\lambda_I}}^\infty\frac {C\, dy''}{(1+|h(y'',x-\lambda \frac {y_2}{\lambda_I}, \frac {y_2}{\lambda_I},\lambda)+\lambda y''|^k)(1+|y''|^h)}|^p_{(x_2,y_2)=\mathcal F_3^{-1}(z_3,\bar z_3)}\nonumber\\
&<&\int_{\mathbb R} dz_{3,I}\int_{\frac {y_2}{\lambda_I}}^\infty\frac {dy''}{1+|y''|^{hp}}\int_{\mathbb R}\frac {C'\,dz_{3,R}}{1+|x-\lambda [\frac {y_2}{\lambda_I}- y'']|^{kp}} |_{(x_2,y_2)=\mathcal F_3^{-1}(z_3,\bar z_3)}\nonumber\\
&<&C''\int_{\mathbb R} dy_{2}\int_{\frac {y_2}{\lambda_I}}^\infty\frac {dy''}{1+|y''|^{hp}}\int_{\mathbb R}\frac {dx}{1+|x-\lambda [\frac {y_2}{\lambda_I}- y'']|^{kp}} \nonumber\\
&=& O(|\lambda_I|)\label{eq:L1_36-q} ,
\eea and hence
\beq
\label{eq:L1_36}
|\alpha(z_3,\bar z_3,\lambda) |_{L^{p}} = O(|\lambda_I|^{\frac 1p}),\ \ 
|p-2|<\epsilon.
\eeq
Using the estimates from \cite{V62} we see, that 
$$
\|\Phi(z_3,\bar z_3,\lambda) - z_3 \|_{L^{\infty}(dz_3d\bar z_3)}=O(|\lambda_I|^{\frac 1p}),
$$
and $\Phi(z_3,\bar z_3,\lambda)$ uniformly converges to $z_3$. So \eqref{E:psi-phi+-reduction} is obtained by composing $\mathcal F_j$, $j=1,\,2,\,3$ and taking $\lambda_I\to 0$.

\end{proof}

\section{The forward problem III: the nonlinear Riemann Hilbert problem}\label{S:NRH-dp}

Theorem \ref{T:existence-cpx-pavlov} yields 
\begin{equation}\label{E:jump-cx}
\Phi^+(x,y,\lambda)-\Phi^-(x,y,\lambda)=-2i\chi^-_I(\varphi(x,y,\lambda),\lambda),\quad\lambda\in\mathbb R,
\end{equation}
where $\chi^-=\chi^-_R+\chi^-_I$. So the Plemelji formula and \eqref{E:lax-complex-reduction} imply
\begin{gather}
\Phi(x,y,\lambda)=x-\lambda y-\frac 1\pi\int_{\mathbb R}\frac{\chi^-_I(\varphi(x,y,\zeta),\zeta)}{\zeta-\lambda}d\zeta,\label{E:cauchy-dp}\\
v(x,y)=-\frac 1{\pi  }\int_{\mathbb R}\chi_I^-(\varphi(x,y,\lambda),\lambda)d\lambda.\label{E:potential-formula-dp}
\end{gather}
From (\ref{E:psi-phi-reduction}) and  (\ref{E:cauchy-dp}), we then derive the non linear integral equation 
\begin{equation}\label{E:nonlinear-inv}
\varphi(x,y,\lambda)+\chi^-_R(\varphi(x,y,\lambda),\lambda)=x-\lambda y-\frac 1\pi\int_R\frac{\chi^-_I(\varphi(x,y,\zeta),\zeta)}{\zeta-\lambda}d\zeta.
\end{equation}In \cite{GSW}, under a small data assumption, (\ref{E:nonlinear-inv}) plays a key role in solving the inverse problem of the Pavlov equation. However, without small data constraints, it is difficult  to study the inverse problem via \eqref{E:nonlinear-inv} since a Fredhlom alternative theorem for the linearization is hard to justify.

On the other hand, 
 it is more direct to adopt the nonlinear Riemann-Hilbert approach \cite{MS3}-\cite{MS7} because  the associated linearized operator is a standard non homogeneous Riemann-Hilbert problem (Section \ref{S:NRH}) with which the Fredhlom alternative theory has been well understood \cite{Ga66}. Therefore,  we will transform the non linear integral equation (\ref{E:nonlinear-inv}) to a nonlinear Riemann-Hilbert problem \eqref{E:NRH-dp} in this section.

\begin{lemma}\label{L:xi-decay}
If $v\in\mathfrak S$ is compactly supported in $y$, then there exists a constant $C_{\lambda,k}$, determined by $\lambda$, $k$, such that
\begin{equation}\label{E:xi-decay}
|\frac{\partial\sigma}{\partial\xi}(\xi,\lambda)|\le \frac {C_{\lambda,k}}{1+|\xi|^k}.
\end{equation}
\end{lemma}
\begin{proof} From \eqref{E:h} and \eqref{E:positive-h-0}, there exist uniform constants $C_i$
\begin{equation}\label{E:h-xi}
\begin{array}{c}
0<C_1\left(|\xi|+1\right)\le |h(y';\xi,-\infty,\lambda)|\le C_2\left(|\xi|+1\right),\\
0<C_1\le |\frac{\partial h}{\partial \xi}(y';\xi,-\infty,\lambda)|\le C_2
\end{array}
\end{equation} For each fixed $\lambda\in\RR$, plugging \eqref{E:h-xi} into \eqref{E:sigma}, and using \eqref{E:h-mu-x}, $v\in\mathfrak S$, $v$ is compactly supported in $y$, we obtain 
\[
\begin{split}
|\frac{\partial\sigma}{\partial\xi}(\xi,\lambda)|
=&\,C\int_{-\infty}^{\infty} |v_{xx}\big(h(y'';\xi,-\infty,\lambda)+ \lambda y'',y''\big)|dy''\\
\le &\,\frac {C_{\lambda,k}}{1+|\xi|^k}.
\end{split}
\]
\end{proof}

\begin{lemma}\label{L:NRH} 
If $ v\in\mathfrak S$  is compactly supported in $y$, then for each fixed $\lambda\in\mathbb R$, the map
\begin{equation}\label{E:NRH-dp-1}
\begin{array}{rl}
\Xi: \mathbb R\quad\longrightarrow \quad \partial\mathfrak D\ &\quad \subset\quad\mathbb C,\\
\Xi(\xi)=\xi+\chi^-(\xi,\lambda),&
\end{array}
\end{equation}diffeomorphically maps $\mathbb R$ to its image, denoted as $\partial\mathfrak D=\partial\mathfrak D(\lambda)$. Namely, $\Xi$ has an  inverse, denoted as $\mathcal M(\zeta,\lambda)$,  satisfying
\begin{equation}\label{E:inverse-map}
\begin{array}{rl}
\mathcal M: \partial\mathfrak D\times\RR\quad\longrightarrow \quad \mathbb R,&\\
\mathcal M(\xi+\chi^-(\xi,\lambda),\lambda)=\xi,&\quad
\forall\xi,\,\forall\lambda\in\mathbb R.
\end{array}
\end{equation}Moreover, denote $\mathfrak D=\partial\mathfrak D\cup \{\xi+\chi(\xi,\lambda)|\ \xi\in\CC^-\}=\{\xi+\chi(\xi,\lambda)|\ \xi\in\overline{\CC^-}\}$. Then there  exist a uniform constant $\hbar>0$, $\tilde {\mathfrak D}=\tilde {\mathfrak D}(\lambda)$, $\tilde {\mathfrak U}(\lambda)$, $\widetilde\chi(\xi,\lambda)$, and $\widetilde{\mathcal M}(\zeta,\lambda)$,  such that 
\beq\label{E:holo-ext}
\begin{split}
(1)\ &\xi+\widetilde\chi(\xi,\lambda):\ \tilde {\mathfrak U}\times \RR \to\ \tilde {\mathfrak D},\quad\quad \overline{\CC^-}\subsetneq\tilde {\mathfrak U}\subset\CC, \ \ \mathfrak D\subsetneq\tilde {\mathfrak D}\subset\CC,\\
& \textit{distance of $\partial{\mathfrak D}(\lambda)$ and $\partial\tilde {\mathfrak D}(\lambda)$ is larger than $\hbar$ for $\forall\lambda\in\RR$},\\
(2)\ &\widetilde\chi(\xi,\lambda)=\chi(\xi,\lambda),\quad \xi\in \overline{\CC^-},\\
(3)\ &\textit{$\widetilde\chi(\xi,\lambda)$ is holomorphic in $\xi\in\tilde {\mathfrak U}$},\\
(4)\ &\textit{the inverse $(\mbox {\bf I}+\chi)^{-1}(\zeta,\lambda)$ on $\mathfrak D$ has a holomorphic }\\
& \textit{extensiton $\widetilde{\mathcal M}(\zeta,\lambda)=(\mbox {\bf I}+\widetilde\chi)^{-1}(\zeta,\lambda)$ on $\tilde {\mathfrak D}$.}
\end{split}
\eeq and
\begin{equation}\label{E:inverse-map-regularity}
\begin{array}{c}
|\partial_{\lambda}^{\nu}\partial_{\xi}^{\mu}\widetilde{\chi}(\xi,\lambda)|_{ L^\infty}\le \frac C{1+|\lambda|^{2+\mu+\nu-\frac 1p}},\ \forall\xi\in \tilde {\mathfrak U},\ \forall\lambda\in\mathbb R,\ \  \mu+\nu\ge 0;\\
|\partial_{\lambda}^{\nu}\partial_{\zeta}^{\mu}\widetilde{\mathcal M}(\zeta,\lambda)|_{ L^\infty}\le \frac C{1+|\lambda|^{\max(0,\mu+\nu-1-\frac 1p)}},\ 
 \forall\zeta\in \widetilde{\mathfrak D},\ \forall\lambda\in\mathbb R,\ \  \mu+\nu>0.
 \end{array}
 \end{equation}
\end{lemma}
\begin{proof}  
Note $\chi$ is the solution of the shifted Riemann-Hilbert problem  \eqref{E:shifted-RH}, 
\begin{equation}\label{E:shifted-RH-0}
\xi+\sigma(\xi,\lambda)+\chi^+(\xi+\sigma(\xi,\lambda),\lambda)=\xi+\chi^-(\xi,\lambda).
\end{equation}
Differentiating both sides of \eqref{E:shifted-RH-0} with respect to $\xi$, one obtains
 a Riemann-Hilbert problem
\begin{equation}\label{E:shifted-RH-2}
\left(1+\frac{\partial \chi}{\partial\xi}\right)^+\left(1+\frac{\partial \sigma}{\partial\xi}\right)
=\left(1+\frac{\partial \chi}{\partial\xi}\right)^-.
\end{equation}
From Lemma \ref{L:sigma}, $\log (1+\frac{\partial \sigma}{\partial\xi})$ is well-defined. Combining with Lemma \ref{L:xi-decay} and Proposition \ref{P:shifted-RH-0-xi-lambda}, one then derives \cite{Ga66}
\begin{equation}\label{E:shifted-RH-3}
1+\frac{\partial \chi}{\partial\xi}(\xi,\lambda)=\exp\left(-\frac 1{2\pi i}\int_{\mathbb R}\frac {\log (1+\frac{\partial \sigma}{\partial\xi})}{t-\xi}dt\right)
\end{equation}
So for $\xi\in\mathbb R$, 
\begin{equation}\label{E:jacobian-c}
1+\frac{\partial \chi^-}{\partial\xi}(\xi,\lambda)=\exp\left(-\frac 1{2\pi i}\int_{\mathbb R}\frac {\log (1+\frac{\partial \sigma}{\partial\xi})}{t-\xi}dt+\frac 12 \log (1+\frac{\partial \sigma}{\partial\xi})\right).
\end{equation}
Therefore, 
\begin{equation}\label{E:jacobian-r}
\begin{array}{c}
\left(1+\frac{\partial \chi^-_R}{\partial \xi}(\xi, \lambda)\right)^2+\left(\frac{\partial \chi^-_I}{\partial \xi}(\xi, \lambda)\right)^2\ne 0, \\
\textit{ for $\forall \xi$, $\forall \lambda\in\mathbb R$.}
\end{array}
\end{equation}
Moreover, applying Proposition \ref{P:shifted-RH-0-xi-lambda}, \eqref{E:jacobian-c}, and \eqref{E:jacobian-r}, we have
\begin{equation}\label{E:jacobian-1}
\begin{array}{c}
0<C_1\le \left(1+\frac{\partial \chi^-_R}{\partial \xi}(\xi, \lambda)\right)^2+\left(\frac{\partial \chi^-_I}{\partial \xi}(\xi, \lambda)\right)^2\le C_2, \\
\textit{ for $\forall \xi$, $\forall \lambda\in\mathbb R$.}
\end{array}
\end{equation}So $\Xi$ maps $\RR$ to its image $\partial\mathfrak D$ diffeomorphically. Combining with Proposition \ref{P:shifted-RH-0-xi-lambda}, we conclude that, for $\forall\lambda\in \RR$, $
\xi+\chi(\xi,\lambda)$ is  a degree one holomorphic map from $\xi\in \overline{\CC^-}$ onto its image  $\mathfrak D=\mathfrak D(\lambda)$. So the existence of $\hbar>0$, $\tilde {\mathfrak D}=\tilde {\mathfrak D}(\lambda)$, $\tilde {\mathfrak U}=\tilde {\mathfrak U}(\lambda)$, $\widetilde\chi(\xi,\lambda)$, and $\widetilde{\mathcal M}(\zeta,\lambda)$,  satisfying \eqref{E:holo-ext} and \eqref{E:inverse-map-regularity}, follows from the inverse function theorem. 
\end{proof}

\begin{theorem}\label{P:NRH-problem}
Suppose $v\in\mathfrak S$ is compactly supported in $y$. There exists a function 
\[
\begin{array}{c}
R(\zeta,\lambda):\widetilde{\mathfrak D}(\lambda)\times \RR\longrightarrow \CC,\\
\textit{$R(\zeta,\lambda)$ is holomorphic in $\zeta\in\widetilde{\mathfrak D}$},
\end{array}
\]with $\widetilde{\mathfrak D}$ defined by \eqref{E:holo-ext}, satisfying {the algebraic constraint
\begin{equation}\label{E;alge}
\begin{array}{c}
\mathcal R(\zeta,\lambda)=\mathcal R(\overline{\mathcal R(\overline\zeta,\lambda)},\lambda),\\
\mathcal R(\zeta,\lambda)=\zeta + R(\zeta,\lambda),
\end{array}
\end{equation}}
and 
the analytical constraint
\begin{gather}
|{\partial_\lambda^\nu}\partial_{\zeta}^{\mu} R(\zeta,\lambda)|_{ L^\infty}\le \frac {C}{1+|\lambda|^{2+\mu+\nu-\frac 1p}},\label{E:R-nonlinear}
\end{gather} such that the nonlinear Riemann-Hilbert problem 
\begin{equation}\label{E:NRH-dp}
\begin{array}{cl}
\Phi^+(x,y,\lambda)=\Phi^-(x,y,\lambda)+R(\Phi^-(x,y,\lambda),\lambda),& \lambda\in\mathbb R,\\
\partial_{\overline\lambda}\Phi(x,y,\lambda)=0, &\lambda\in\CC^\pm,\\
\Phi(x,y,\lambda) -(x-\lambda y)\to 0,&|\lambda|\to\infty
\end{array}
\end{equation}
holds. Moreover, 
the representation formula for the complex eigenfunction and the potential are
\begin{gather}
\Phi(x,y,\lambda)=x-\lambda y+\frac 1{2\pi i}\int_{\mathbb R}\frac{ R(\Phi^-(x,y,\lambda'),\lambda')}{\lambda'-\lambda}d\lambda',\label{E:cauchy-dp-nrh}\\
v(x,y)=\frac 1{2\pi i}\int_{\mathbb R} R(\Phi^-(x,y,\lambda),\lambda)d\lambda.\label{E:potential-formula-dp-nrh}
\end{gather}
\end{theorem}
\begin{proof} Via \eqref{E:psi-phi+-reduction}, \eqref{E:inverse-map}, one has
\[
\varphi(x,y,\lambda)={\mathcal M}(\Phi^-(x,y,\lambda),\lambda).
\]
Moreover, if we define
\begin{equation}\label{E:inverse-r}
 R(\zeta,\lambda)=-2i\widetilde\chi_I(\widetilde{\mathcal M}(\zeta,\lambda),\lambda)\quad\zeta\in\widetilde{\mathfrak D},\ \lambda\in\RR
\end{equation}
with $\widetilde{\mathfrak D}$, $\widetilde\chi=\widetilde\chi_R+i\widetilde\chi_I$, and $\widetilde{\mathcal M}$ defined by \eqref{E:holo-ext}, then Theorem \ref{T:existence-cpx-pavlov}, \eqref{E:jump-cx} imply \eqref{E:NRH-dp}. Therefore, the analytical constraint \eqref{E:R-nonlinear} can be justified by \eqref{E:inverse-map-regularity}  and \eqref{E:inverse-r}. 
Besides,  from \eqref{E:NRH-dp} and the reality condition \eqref{eq:L1_1-sym-reduction}, we have
\beq\label{E:aa}
\begin{split}\,&\Phi^+(x,y,\lambda)\\
=\,&\mathcal R(\Phi^-(x,y,\lambda),\lambda)\\
=\,&\mathcal R(\overline{\Phi^+(x,y,\lambda)},\lambda)\\
=\,&\mathcal R(\overline{\mathcal R(\Phi^-(x,y,\lambda),\lambda)},\lambda)\\
=\,&\mathcal R(\overline{\mathcal R(\overline{\Phi^+(x,y,\lambda)},\lambda)},\lambda),\quad \forall\lambda\in\RR.
\end{split}
\eeq Thus  {the algebraic  constraint  \eqref{E;alge} can be derived from the holomorphicy of $R(\zeta,\lambda)$ (from that of $\widetilde\chi(\cdot,\lambda)$, $\widetilde {\mathcal M}(\cdot,\lambda)$) and  \eqref{E:aa}.}

Equations (\ref{E:cauchy-dp-nrh}) is immediate from the Plemelji formula and the jump formula (\ref{E:NRH-dp}). Equation (\ref{E:potential-formula-dp-nrh}) can be derived from 
\begin{eqnarray}
0&=&\left[\partial_y+(\lambda+v_x(x,y))\partial_x\right]\Phi(x,y,\lambda)\nonumber\\
&=& v_x(x,y)+\frac 1{2\pi i}\int_{\mathbb R}\frac {(\partial_y+(\lambda'+v_x(x,y))\partial_x) R(\Phi^-(x,y,\lambda'),\lambda')}{\lambda'-\lambda}d\lambda'\nonumber\\
&&-\frac 1{2\pi i}\left(\partial_{x}\int_{\mathbb R} R(\Phi^-(x,y,\lambda'),\lambda')d\lambda'\right)\nonumber\\
&=&v_x(x,y)-\frac1{2\pi i}\left(\partial_{x}\int_{\mathbb R} R(\Phi^-(x,y,\lambda'),\lambda')d\lambda'\right)\nonumber.
\end{eqnarray}
\end{proof}

\begin{definition}\label{D:scattering-data}
If $v\in\mathfrak S$ is compactly supported in $y$, then the \textit{\textbf{spectral data}} of $v(x,y)$ is defined by
\beq\label{E:scattering}
 R(\zeta,\lambda)=-2i\widetilde\chi_I(\widetilde{\mathcal M}(\zeta,\lambda),\lambda),\quad\zeta\in\tilde {\mathfrak D},\ \lambda\in\RR,
\eeq where $\widetilde\chi=\widetilde\chi_R+i\widetilde\chi_I$, $\widetilde{\mathcal M}$, and $\tilde {\mathfrak D}$ are defined by \eqref{E:holo-ext}.
\end{definition}

\begin{remark}
From \eqref{E:NRH-dp} and \eqref{E:scattering}, the spectral data for the Pavlov equation is purely continuous and nonlinear, unlike the $\bar\partial$-problem formulated for general  soliton equations \cite{ABF}, \cite{BC81}, \cite{Wu1}.
\end{remark}

{\begin{remark}\label{R:obstruction-1}
The holomorphicy condition on the spectral data $R(\cdot,\lambda)$ is crucial for the reality condition \eqref{E;alge}  and is indispensable for the nonlinear Riemann-Hilbert approach in Section  \ref{S:lax-cauchy}. 
\end{remark}}

\begin{proposition}\label{P:NRH-problem-ext}
Suppose $0\ne v\in\mathfrak S$ is compactly supported in $y$. There exist uniform   constants $\epsilon_0$, $\delta_0$, 
 \beq\label{E:obstruction}
0< \epsilon_0,\  \delta_0<\infty,
 \eeq  
such that  
 if 
\beq\label{E:omega-a}
\begin{array}{c}
|\omega^-(x,y,t,\lambda)-\left(\Phi^-(x,y,\lambda)-x+\lambda y\right)|_{H^p(\RR,d\lambda)}\le \delta_0,\\
 0\le t\le \epsilon_0,
\end{array}
\eeq then
\[
\begin{array}{c}
\omega^-(x,y,t,\lambda)+x-\lambda y-\lambda^2 t\in\widetilde{\mathfrak D}(\lambda),\\
\forall  \lambda\in\RR,\ \  0\le t\le \epsilon_0.
\end{array}
\]

\end{proposition}

\begin{proof} Let $\hbar$ be defined by \eqref{E:holo-ext}.
Proposition \ref{P:shifted-RH-0-xi-lambda} implies there exists $N$ such that   
\beq\label{E:pro}
\sup_{\zeta\in\partial\mathfrak D } |\zeta_I|\le\frac {\hbar}4,\quad\textit{for  $\forall\ |\lambda|>N$}.
\eeq
 Denote $\Im(\zeta)=\zeta_I$. Hence \eqref{E:omega-a} and \eqref{E:pro} imply, for $\lambda\in\RR$, $|\lambda|>N$,
\beq\label{E:pro-1}
\begin{split}
&|\Im\left(\omega^-(x,y,t,\lambda)+x-\lambda y-\lambda^2 t\right)|\\
=&|\Im\left(\omega^-(x,y,t,\lambda)+x-\lambda y\right)|\\
\le&|\Im\left[\omega^-(x,y,t,\lambda)-\left(\Phi^-(x,y,\lambda)-(x-\lambda y)\right)\right]|+|\Im \Phi^-(x,y,\lambda))|\\
=&|\Im\left[\omega^-(x,y,t,\lambda)-\left(\Phi^-(x,y,\lambda)-(x-\lambda y)\right)\right]|+|\Im \Xi(\varphi(x,y,\lambda),\lambda))|\\
\le&|\Im\left[\omega^-(x,y,t,\lambda)-\left(\Phi^-(x,y,\lambda)-(x-\lambda y)\right)\right]|+\sup_{\zeta\in\partial\mathfrak D } |\zeta_I|\\
\le &C\delta_0+\frac {\hbar}4.
\end{split}
\eeq Here $C$ is a uniform constant determined by
 Sobolev's theorem. So  if 
 \[\delta_0 =\frac {\hbar}{2C},
 \] then \eqref{E:holo-ext}, \eqref{E:pro}, and \eqref{E:pro-1} yield
\beq\label{E:nrh-ext-large}
\begin{array}{c}
\omega^-(x,y,t,\lambda)+x-\lambda y-\lambda^2 t\in \tilde {\mathfrak D},\\
 \textit{for }\ \,t\le\epsilon_0,\ \forall\lambda\in\RR,\ |\lambda|>N.
\end{array}
\eeq 

On the other hand, for $|\lambda|\le N$. Since 
\[
\Phi^-(x,y,\lambda)=\varphi(x,y,\lambda)+\chi^-(\varphi(x,y,\lambda),\lambda)\subset\partial\mathfrak D.
\]
We have, if
\begin{equation}\label{E:nrh-ext-small}
|\omega^-(x,y,t,\lambda)+x-\lambda y-\lambda^2 t-\Phi^-(x,y,\lambda)|_{L^\infty}\le\hbar,
\end{equation}
then 
\beq\label{E:nrh-ext-small-1}
\begin{array}{c}
\omega^-(x,y,t,\lambda)+x-\lambda y-\lambda^2 t\ \ \in\ \ \tilde {\mathfrak D}
\end{array}
\eeq
by \eqref{E:holo-ext}. Condition \eqref{E:nrh-ext-small} can be assured by \eqref{E:omega-a}, $0\le t\le\epsilon_0$, $|\lambda|\le N$, and 
\[\epsilon_0 = \frac{\hbar}{2N^2}.\]

 Finally, $0<\epsilon_0,\ \delta_0<\infty$ follows from $0<\hbar<\infty$ which is implied by Proposition \ref{P:shifted-RH-0-xi-lambda}, \eqref{E:jacobian-1}, $v\ne 0$, and Liouville's theorem.
\end{proof} 

The following lemma will be used to compute the index of the linearized Riemann-Hilbert problem \cite{Ga66} for the inverse problem.
\begin{lemma}\label{L:compact-y}
Suppose $v\in\mathfrak S$ and is compactly supported in $y$. Then 
\begin{equation}\label{E:index}
\begin{split}
&\ \,\mbox{Ind}\,\left(1+\frac{\partial R}{\partial \zeta}(\Phi^-(x,y,\lambda),\lambda) \right)\\
\equiv&\ \frac 1{2\pi i}\int_{\mathbb R}\frac{\partial_\lambda{\partial _\zeta}R(\Phi^-(x,y,\lambda),\lambda)}{1+\frac{\partial R}{\partial\zeta}(\Phi^-(x,y,\lambda),\lambda)}d\lambda \\
=&\ 0.
\end{split}
\end{equation}

\end{lemma}
\begin{proof} A direct computation yields
\beq\label{E:index-zero-2}
\begin{split}
&\ |1+\frac{\partial R}{\partial \zeta}(\Phi^-(x,y,\lambda),\lambda)|\\
=&\ |1-2i\partial_\zeta\left[ \chi^-_I\left({ \mbox {\bf I}}+\chi^-\right)^{-1}\right](\Phi^-(x,y,\lambda),\lambda|)\\
=&\ |1-2i\frac{\partial\chi^-_I}{\partial\xi}
\frac{\partial\xi}{\partial\zeta}
(\Phi^-(x,y,\lambda),\lambda|)|\\
=&\ |1-\frac{2i\frac{\partial\chi^-_I}{\partial\xi}}{
1+\frac{\partial\chi^-}{\partial\xi}}
(\varphi(x,y,\lambda),\lambda)|\\
=&\ |\frac{1+\overline{\partial_\xi\chi^-}}{1+{\partial_\xi\chi^-}}(\varphi(x,y,\lambda),\lambda)|\\
=&\ 1.
\end{split}
\eeq  Therefore $\mbox{Ind}\,\left(1+\frac{\partial R}{\partial \zeta}(\Phi^-(x,y,\lambda),\lambda) \right)$  is well defined for $x$, $y,\,\lambda\in\RR$. 

Furthermore, noting that the index function \cite{Ga66} is continuous and integer valued and using \eqref{E:index-zero-2}, it 
suffices to show 
\begin{equation}\label{E:index-s-y}
|\partial_\lambda\partial_\zeta R(\Phi^-(x,y,\lambda),\lambda)|_{L^1(\mathbb R,d\lambda)}\ll 1,\ \textit{for $y=0$, and $x\to \infty$} .
\end{equation}
To verify \eqref{E:index-s-y}, we first note 
\begin{equation}\label{E:dc}
\textit{$\partial_\lambda\partial_\zeta R(\Phi^-(x,y,\lambda),\lambda)$ are uniformly bounded by an $L^1(\mathbb R,d\lambda)$ function}.
\end{equation} from  \eqref{E:R-nonlinear}. Moreover, for arbitrary positive constant $M$, $|\lambda|<M$,   as $x\to\infty$, 
\begin{equation}\label{E:denominator}
\begin{split}
&\ |\partial_\lambda\partial_\zeta R(\Phi^-(x,0,\lambda),\lambda)|\\
\le&\ C|\partial_\lambda\partial_\xi\chi^-_I(\varphi(x,0,\lambda),\lambda)|\\
= &\ C|\partial_\lambda\partial_\xi\chi^-_I(x+\varphi_0(x,0,\lambda),\lambda)| \\
\to &\ 0
\end{split}
 \end{equation}by \eqref{E:h}, \eqref{E:phi}, Lemma \ref{L:NRH}, and Proposition \ref{P:shifted-RH-0-xi-lambda}.  
Therefore \eqref{E:index-s-y} is justified by \eqref{E:dc}, \eqref{E:denominator}, the Lebesque's Dominated Convergence Theorem, Proposition \ref{P:shifted-RH-0-xi-lambda}, and taking $M$ sufficiently large.
\end{proof}

The following lemma is from \cite{GSW} (cf. Proposition 3.5 and its proof therein) and will be used in Theorem \ref{T:inverse-problem-P-1-S}. 
\begin{lemma}\label{L:decay}
Suppose $v\in\mathfrak S$ and is compactly supported. Then for fixed $t>0$ and $\omega(x,y,t,\lambda)\in L^\infty(\RR,d\lambda)$,  as $\lambda\to\infty$,
\beq\label{E:decay}
|\partial_\xi^\mu\chi^-(\omega(x,y,t,\lambda)+x-\lambda y-\lambda^2 t,\lambda)|\le  \mathcal O\left(\frac{1}{1+|\lambda|^{2\mu+3-\frac 1p}}\right).
\eeq
\end{lemma}

\section{The inverse problem I: the nonlinear Riemann-Hilbert problem}\label{S:NRH}
We turn to the inverse problem of the Pavlov equation. 
Recall that for $t=0$, we have
\begin{equation}\label{E:Manakov-Santini-0}
\begin{split}
\Phi^+(x,y,\lambda)=&\ \Phi^-(x,y,\lambda)+R(\Phi^-(x,y,\lambda),\lambda) \\
\equiv &\ \mathcal R(\Phi^-(x,y,\lambda),\lambda).
\end{split}
\end{equation} Our goal is to solve the nonlinear Riemann-Hilbert problem
\begin{equation}\label{E:Manakov-Santini-NRH}
\begin{array}{cl}
\Psi^+(x,y,t,\lambda)
=\mathcal R(\Psi^-(x,y,t,\lambda),\lambda),&\quad \lambda\in\RR, \ t>0,\\
\partial_{\bar\lambda}\Psi(x,y,t,\lambda)= 0,&\quad \lambda\in\CC^\pm,\ t>0,\\
\Psi(x,y,0,\lambda)=\Phi(x,y,\lambda).&
\end{array}
\end{equation}
To tackle the nonlinear Riemann-Hilbert problem, we will adopt a nonlinear Newtonian iteration approach \cite{C06}.  More precisely,  first
normalize 
\begin{equation}\label{E:renormalization}
\Psi(x,y,t,\lambda)=\omega(x,y,z,t)+x-\lambda y-\lambda^2t,
\end{equation} and let
\beq
\begin{split}
&{\mathcal R}(\omega^-+x-\lambda y-\lambda^2t,\lambda)- (\omega^++x-\lambda y-\lambda^2 t)\\
=&R(\omega^-+x-\lambda y-\lambda^2t,\lambda)+\omega^-- \omega^+\\
=\ \,&{\mathcal G}( \omega^+,\omega^-)
.
\end{split}\label{E:FP-ite}
\eeq
Define $\omega_n^\pm(x,y,t,\lambda)$ recursively by
\begin{equation}\label{E:recursive-formula}
{ \textbf{G}}_n\delta\Omega_{n+1}={ \textbf{G}}_n\delta\Omega_n-\mathcal G(\Omega_n),\quad\textit{for $n\ge 0$, }
\end{equation}
with
\bea
&&\Omega_{n}=(\omega_n^+,\omega_n^-),\label{E:recursive-small-ome-1}\\
&&\omega_0^\pm(x,y,t,\lambda)\equiv\Phi^\pm(x,y,\lambda)-(x-\lambda y),\label{E:recursive-small-ome-2}\\
&&\delta\Omega_n=\Omega_n-\Omega_{0}=(\delta\omega^+_n,\delta\omega^-_n).\label{E:recursive-small-ome-4}
\eea
and $\textbf{G}_n$ is the linearization of $\mathcal G$ at $\Omega_n$, i.e.,
\begin{equation}\label{E:linearization-G}
\begin{split}
&\textbf{G}_n\delta\Omega\\
=&-\delta\omega^++\frac{\partial\mathcal R}{\partial\zeta}|_{\omega^-_n+x-\lambda y-\lambda^2 t}\delta \omega^-\\
=&-\delta\omega^++\delta \omega^-+\frac{\partial R}{\partial\zeta}(\omega_n^-+x-\lambda y-\lambda^2 t,\lambda)\delta \omega^-\\
\equiv&-\delta\omega^++\mathcal J_{n}\delta \omega^-.
\end{split}
\eeq

Thus  \eqref{E:recursive-formula} can be written as the non homogeneous Riemann-Hilbert problem
\begin{equation}\label{E:linearization-NRH}
\begin{split}
\delta\omega_{n+1}^+
=&\ \  \mathcal J_n\ \delta\omega_{n+1}^--\left({\textbf G}_n\delta\Omega_n-\mathcal G(\Omega_n)\right).
\end{split}
\end{equation}

Owing to Proposition \ref{P:NRH-problem-ext}, we will justify 
\beq\label{E:assumption}
\delta\omega_n^-\le \delta_0\ \ \ \textit{for }\ \ 
0\le t\le \epsilon_1\le\epsilon_0
\eeq to make the above algorithm eligible. Hence  from (\ref{E:FP-ite})-(\ref{E:linearization-G}),  it is easy to see that if  $\{\Omega_n\}$ converge to $\Omega$ in $ H^p(\RR,d\lambda)$, then  ${\mathcal G}(\Omega)$ is well-defined and ${\mathcal G}(\Omega)=0$. Combining with the Cauchy integral formula, consequently, we will obtain a solution for the nonlinear Riemann-Hilbert problem (\ref{E:Manakov-Santini-NRH}). 

\begin{lemma} {\rm{\bf(Estimate of ${\textrm{\bf  G}}_{n}^{-1}$)}} \label{L:gn-inv-S} Suppose $p>1$ and $v\in\mathfrak S$ is compactly supported in $y$. If 
\eqref{E:assumption} is valid, $ g(\lambda)\in H^p(\RR,d\lambda)$, then the non homogeneous Riemann-Hilbert problem
\beq\label{E:linearize-RH-S}
\begin{array}{cl}
 f^+=\mathcal J_n f^-+ g,&\quad\lambda\in\RR,\\
 \partial_{\bar\lambda}f=0,&\quad\lambda\in\CC^\pm.
 \end{array}
\eeq
admits a unique solution $ f$ such that 
\begin{equation}\label{E:gn-inv-1-S} 
| f^\pm|_{H^p(\RR,d\lambda)}  
\le C_{\omega_0}| g |_{H^p(\RR,d\lambda)} .
\end{equation}Here  $C_{\omega_0}$ is a constant determined by  $|\omega_0^\pm|_{H^p(\mathbb R,d\lambda)}$.
\end{lemma}
\begin{proof} Using \eqref{E:index-zero-2} and \eqref{E:assumption} (shrink $\hbar$ if necessary) and adapting the proof of Lemma \ref{L:compact-y},  one can prove $\textit{Ind}\,\mathcal J_n$ is well defined and $\textit{Ind}\,\mathcal J_n=0$.  Hence one has the unique solvability of the Riemann-Hilbert problem 
{\begin{equation}\label{E:non-homo}
\begin{array}{cl}
X^+(\lambda)=\mathcal J_nX^-(\lambda),&\lambda\in\RR,\\
\partial_{\overline\lambda}X(\lambda)=0, &\lambda\in\CC^\pm,\\
\textit{$X^\pm\to 1$},&|\lambda|\to\infty,\\
{}\\
\partial_\lambda X^\pm,\,\partial_\lambda {X^\pm}^{-1}\in H^p(\RR,d\lambda),\\
X^\pm,\,{X^\pm}^{-1}\in L^\infty,
\end{array}
\end{equation}}
and of the non homogeneous Riemann-Hilbert problem \eqref{E:linearize-RH-S}  \cite{Ga66}
\beq\label{E:solution-diff} 
f=X\left(\frac{1}{2\pi i}\int_{\RR}\frac{\frac g{X^+}}{\zeta-\lambda}d\zeta\right).
\eeq
So \eqref{E:gn-inv-1-S} follows.
\end{proof}


\begin{lemma} {\bf (Small time estimate and regularities)}\label{L:local-small} 
If $p>2$ and $v\in\mathfrak S$ is compactly supported in $y$, then
\begin{equation}\label{E:basic-small}
\begin{split}
&|\mathcal G(\omega_0^\pm)|_{H^p}<C_{\omega_0}(1+|y|)t, \ \textit{for $0\le t\le \epsilon_0$}.
\end{split}
\end{equation}
Moreover, if \eqref{E:assumption} is valid,
 then
\begin{gather}
|\partial_\zeta^\mu{\mathcal R}(\omega^-+x-\lambda y-\lambda^2 t,\lambda)|_{H^p(\RR,d\lambda)}<C_{\omega_0}(1+|y|), \ \textit{ $0\le t\le \epsilon_1$}.\label{E:basic-linear}
\end{gather}Here $C_{\omega_0}$ is determined by $|\omega_0^\pm(x,y,\lambda)|_{H^p(\RR,d\lambda)}$.
\end{lemma}
\begin{proof} 

 Formula (\ref{E:Manakov-Santini-0}) and (\ref{E:FP-ite}) imply 
\begin{equation}\label{E:initial}
\mathcal R(\omega_0^-+x-\lambda y,\lambda)-(\omega_0^++x-\lambda y)=0.
\end{equation}
Hence
	\begin{eqnarray}
	&&|\mathcal G(\omega_0^\pm)|_{H^p}\nonumber\\
	&=&|\mathcal R(\omega_0^-+x-\lambda y-\lambda^2 t,\lambda)-(\omega_0^++x-\lambda y-\lambda^2 t)|_{H^p}\nonumber\\
	&\le &|\mathcal R(\omega_0^-+x-\lambda y,\lambda)-(\omega_0^++x-\lambda y)|_{H^p}\nonumber\\
	&&+|\mathcal R(\omega_0^-+x-\lambda y-\lambda^2 t,\lambda)-\mathcal R(\omega_0^-+x-\lambda y,\lambda)+\lambda^2 t|_{H^p}\nonumber\\
	&= &| R(\omega_0^-+x-\lambda y-\lambda^2 t,\lambda)- R(\omega_0^-+x-\lambda y,\lambda)|_{H^p}\nonumber
		.\label{E:g-0-estimate}
	\end{eqnarray} 
	Let
	 $\zeta_t=\omega_0^-+x-\lambda y-\lambda^2 t$. 
	Therefore, to prove \eqref{E:basic-small}, it yields to showing
	\bea
	|R(\zeta_0,\lambda)-R(\zeta_t,\lambda)|_{L^p} &\le & C(1+|y|)t,\label{E:small-1}\\
	|\partial_\lambda R(\zeta_0,\lambda)-\partial_\lambda R(\zeta_t,\lambda)|_{L^p}
	&\le & C(1+|y|)t.\label{E:small-2}
  \eea

Theorem \ref{P:NRH-problem} and the mean value theorem (or inequality) imply 
  \bea
  &&|\partial_\lambda R(\zeta_0,\lambda)-\partial_\lambda R(\zeta_t,\lambda)|_{L^p}\nonumber\\
&\le & | \partial_\zeta R(\zeta,\lambda)|_{\zeta=\xi_0}\frac{\partial \zeta_0}{\partial\lambda}-\partial_\zeta R(\zeta,\lambda)|_{\zeta=\zeta_t}\frac{\partial \zeta_t}{\partial\lambda}|_{L^p}\nonumber\\
&&+|\partial_\lambda R(\zeta,\lambda)|_{\zeta=\xi_0}-\partial_\lambda R(\zeta,\lambda)|_{\zeta=\zeta_t}|_{L^p}\nonumber\\	
&\le & | \partial_\zeta R(\zeta,\lambda)|_{\zeta=\xi_0}\frac{\partial \zeta_0}{\partial\lambda}-\partial_\zeta R(\zeta,\lambda)|_{\zeta=\xi_0}\frac{\partial \zeta_t}{\partial\lambda}|_{L^p}\nonumber\\
&&+|\partial_\zeta R(\zeta,\lambda)|_{\zeta=\xi_0}\frac{\partial \zeta_t}{\partial\lambda}-\partial_\zeta R(\zeta,\lambda)|_{\zeta=\zeta_t}\frac{\partial \zeta_t}{\partial\lambda}|_{L^p}\nonumber\\
&&+|\partial_\lambda R(\zeta,\lambda)|_{\zeta=\xi_0}-\partial_\lambda R(\zeta,\lambda)|_{\zeta=\zeta_t}|_{L^p}\nonumber\\
&\le & | \partial_\zeta R(\zeta,\lambda)|_{\zeta=\xi_0}2\lambda t|_{L^p}+C_{\omega_0}|\sup_{0\le t'\le 1}\partial_\zeta^2 R(\zeta_{t'},\lambda)t\lambda^2(1+|y|+2|\lambda| t)|_{L^p}\nonumber\\
&&+|\sup_{0\le t'\le 1}\partial_\lambda\partial_\zeta R(\zeta_{t'},\lambda)\lambda^2 t|_{L^p}\nonumber\\
&\le & C_{\omega_0}t\left|\frac{1+|y|}{1+|\lambda|^{1-\frac 1p}}\right|_{L^p}\nonumber\\
&\le& C_{\omega_0}(1+|y|)t,\label{E:small-2-1}
\eea
where the constant $C_{\omega_0}$ is determined by $|\omega_0(x,y,\cdot)|_{H^p}$. Note we have used  \eqref{E:assumption} valid for $\omega_0$) and $p>2$ in the above derivation. Thus (\ref{E:small-2})  is proved. Estimate \eqref{E:small-1} can be proved similarly.

 Similarly, \eqref{E:basic-linear} can be derived by
	\begin{eqnarray}
	&&|\partial_\zeta^\mu{\mathcal R}(\omega^-+x-\lambda y-\lambda^2 t,\lambda)|_{H^p}\nonumber\\
&\le & C_{\omega_0}(\sum_{k=0}^1|\partial_\zeta^{\mu+k} R(\omega^-+x-\lambda y-\lambda^2 t,\lambda)(1+|y|+|\lambda|t)^k|_{L^p}\nonumber\\
&&+|\partial_\zeta^\mu \partial_\lambda R(\omega^-+x-\lambda y-\lambda^2 t,\lambda)|_{L^p})\nonumber\\
	&\le& C_{\omega_0}\left|\frac {1+|y|+t}{1+|\lambda|^{2+\mu-\frac 1p}}\right|_{L^p}\nonumber\\
	&\le& C_{\omega_0}(1+|y|+t).\nonumber
	\end{eqnarray}

\end{proof}

\begin{lemma} \label{L:quadratic} {\rm{\bf(Uniform boundedness)} }
 Suppose  $p>2$ and $v\in\mathfrak S$ is compactly supported in $y$. There exist  uniform constants $\epsilon_1=\epsilon_1(y,p,\omega_0^\pm)\le \epsilon_0$ and $ \delta_1=\delta_1(y,p, \omega_0^\pm)\le \delta_0$ such that 
\begin{equation}\label{E:uniform}\begin{split}
&\hskip.5in|\delta \Omega_n|_{H^p}=|\delta \omega_n^+|_{H^p}+|\delta \omega_n^-|_{H^p}\le\delta_1,\quad \ \forall n,\quad\textit{for  $0\le t\le \epsilon_1$. }
\end{split}
\end{equation}
\end{lemma}
\begin{proof} If
\begin{equation}\label{E:bounded-assumption}
t\le\epsilon_1\le \epsilon_0,\quad|\delta\Omega_n|_{H^p}\le \delta_1\le \delta_0,
\end{equation}then   Lemma \ref{L:gn-inv-S},  \ref{L:local-small}, Theorem \ref{P:NRH-problem}, and the mean value inequality imply
\begin{eqnarray}
&&|\delta \Omega_{n+1}|_{H^p}\nonumber\\
&\le &  C_{\omega_0}|- {\mathcal G}(\Omega_n^\pm)+{\textrm{\bf  G}}_{n}\delta \Omega_n |_{H^p}\nonumber\\
&=& C_{\omega_0}|-\mathcal G(\Omega_0)+ \mathcal G(\Omega_0) - {\mathcal G}(\Omega_n)+{\textrm{\bf  G}}_{n}\delta \Omega_n |_{H^p}\nonumber\\
&=&C_{\omega_0}|- {\mathcal G}(\Omega_0)+{\textrm{\bf  G}}_{n}\delta \Omega_n - \left[\int_0^1\partial_\zeta\mathcal G(\theta\Omega_0+(1-\theta)\Omega_n)d\theta \right] \delta\Omega_n|_{H^p}\nonumber\\
&\le &C_{\omega_0}(| {\mathcal G}(\Omega_0)|_{H^p}+\sup_{\zeta\in\tilde{\mathfrak D}}|\partial_\zeta^2\mathcal R(\zeta,\lambda)|\ |\left(\delta \omega_n^\pm\right)^2|_{H^p})\nonumber\\
&=& C_\ast(t+\delta^2_1)\label{E:c-star}
\end{eqnarray}
where $C_{\omega_0}$ is a constant determined by $|\omega_0^\pm(x,y,\cdot)|_{H^p}$ and $C_\ast=C_{\omega_0}(1+|y| )$. Applying an induction argument, condition (\ref{E:bounded-assumption}) and the lemma can be proved by successively choosing 
\begin{equation}\label{E:tilde-c}
\delta_1=\min\{\frac 1{2C_\ast},\,\delta_0\}, \quad \epsilon_1=\min\{\frac{\delta_1}{2C_\ast},\,\epsilon_0\}.
\end{equation}
\end{proof}
Hence the recursive formula in (\ref{E:FP-ite})-(\ref{E:linearization-G}) are eligible by Proposition \ref{P:NRH-problem-ext} and Lemma \ref{L:quadratic}.  Moreover,

\begin{lemma} \label{L:convergence-0} {\rm{\bf(Convergence)} }
Suppose  $p>2$ and $v\in\mathfrak S$ is compactly supported in $y$. Let $\epsilon_1$ and $\delta_1$ be defined by (\ref{E:tilde-c}).  Then, for $0\le t\le \epsilon_1$,  $\{\omega_n^\pm\}$ converges to some $\omega^\pm(x,y,t,\lambda)$, with $|\omega^\pm-\left(\Phi^\pm-(x-\lambda y)\right)|_{H^p}\le \delta_1$.
\end{lemma}
\begin{proof} Let 
\[
\Delta\Omega_{n+1}=\Omega_{n+1}-\Omega_n=(\omega_{n+1}^+-\omega_n^+,\ \omega_{n+1}^--\omega_n^-).
\] From (\ref{E:FP-ite})-(\ref{E:linearization-G}), one has
\begin{equation}\label{E:rh-0}
{\textrm{\bf  G}}_n\Delta\Omega_{n+1} 
={\textrm{\bf  G}}_{n-1} \Delta\Omega_{n}-\left(\mathcal G(\Omega_n)-\mathcal G(\Omega_{n-1})\right).
\end{equation}
Similarly, we obtain
\begin{equation}\label{E:conv-1}
| -{\textrm{\bf  G}}_{n-1}\Delta\Omega_n+\mathcal G(\Omega_n)-\mathcal G(\Omega_{n-1})|_{H^p}\le C_{\omega_0}(1+|y|)|\Delta\Omega_n|^2_{H^p}.
\end{equation}
Here $C_{\omega_0}$ is a constant determined by $|\omega_0^\pm(x,y,\cdot)|_{H^p}$. Applying Lemma \ref{L:gn-inv-S} and \ref{L:local-small},  
we derive
\begin{equation}\label{E:c-dagger}
|\Delta\Omega_{n+1}|_{H^p}\le C_\dagger|\Delta\Omega_n|^2_{H^p}
\end{equation}
from (\ref{E:rh-0}), (\ref{E:conv-1}). Here $C_\dagger=C_{\omega_0}'(1+|y|)$  is a  constant independent of $n$. Thus by Lemma \ref{L:quadratic},
\[
|\Delta\Omega_{n+1}|_{H^p}\le C_\dagger\delta_1|\Delta\Omega_n|_{H^p}.
\]
For small $\delta_1$ satisfying $2C_\dagger\delta_1<1$ (redo the proof of Lemma \ref{L:quadratic} if necessary), we obtain
\[
|\Delta\Omega_{n+1}|_{H^p}\le\frac 12|\Delta\Omega_n|_{H^p},
\]
which implies the convergence of $\{\Omega_n\}$.
\end{proof}

\begin{theorem} \label{T:uniqueness} {\bf (Local solvability for the NRH problem)} 
For  $p>2$, $\forall x$, $\forall y\in\mathbb R$, and  $v\in\mathfrak S$ is compactly supported in $y$. Let $\epsilon_1=\epsilon_1(y, p,\omega_0^\pm)$, $\delta_0$ 
be defined by (\ref{E:tilde-c}) and Proposition \ref{P:NRH-problem-ext}. Let
\beq\label{E:time}
\begin{split}
\epsilon_2=\min\{\epsilon_1(y, p,\omega_0^\pm), \epsilon_1(y, 2p,\omega_0^\pm)\}.
\end{split}
\eeq  Then the nonlinear Riemann-Hilbert problem 
\begin{equation}\label{E:Manakov-Santini-NRH-local}
\begin{array}{cl}
\Psi^+(x,y,t,\lambda)= \mathcal R(\Psi^-(x,y,t,\lambda),\lambda),&\ \lambda\in\RR,\  0\le t\le \epsilon_2,\\
\partial_{\bar\lambda}\Psi(x,y,t,\lambda)=0,&\ \lambda\in\CC^\pm,\  0\le t\le  \epsilon_2,\\
\Psi(x,y,0,\lambda)= \Phi(x,y,\lambda)&\\
\end{array}
\end{equation}  has a unique solution  satisfying the reality condition
\begin{equation}\label{E:psi-reality}
\Psi^+(x,y,t,\lambda)=\overline{\Psi^-(x,y,t,\lambda)}
\end{equation} and
\begin{equation}\label{E:reg}
\begin{array}{c}
|\Psi^\pm(x,y,t,\lambda)-\left(\Phi^\pm(x,y,\lambda)-\lambda^2 t\right)|_{H^p}<\delta_0,\\
\partial_t^h\partial_y^k\partial_x^\mu[\Psi^\pm-( x-\lambda y-\lambda^2 t)]\in L^p,\\
 0\le h+k+\mu\le 2,\ h<2.
 \end{array}
\end{equation} 
\end{theorem}

\begin{proof} Therefore, for each $x,y\in\mathbb R$, from (\ref{E:FP-ite})-(\ref{E:linearization-G}), Lemma \ref{L:convergence-0}, Plemelji formula, and
\[
\Psi(x,y,t,\lambda)=x-\lambda y-\lambda^2 t+\frac 1{2\pi i}\int_{\mathbb R}\frac {R(\omega^-+x-\lambda'y-{\lambda'}^2t,\lambda')}{\lambda'-\lambda}d\lambda', \] we obtain a local
solution for the nonlinear Riemann-Hilbert problem \eqref{E:Manakov-Santini-NRH-local}. 

The reality condition in \eqref{E:psi-reality} follows from the algebraic constraint \eqref{E;alge} (adapting \eqref{E:aa} for the extension $\mathcal R$). Suppose 
\beq\label{E:uniq}
\begin{array}{c}
\Psi^+_i(x,y,t,\lambda)= \mathcal R(\Psi^-_i(x,y,t,\lambda),\lambda),\\
\Psi^\pm_i(x,y,0,\lambda)= \Phi^\pm(x,y,\lambda),\\
|\Psi^\pm_i(x,y,t,\lambda)-\left(\Phi^\pm(x,y,\lambda)-\lambda^2 t\right)|_{H^p}<\delta_0
\end{array}
\eeq for $i=1,\,2$. Denote 
\[
\begin{array}{c}
\Pi^\pm(x,y,t,\lambda)=\Psi^\pm_1(x,y,t,\lambda)-\Psi^\pm_2(x,y,t,\lambda),\\
\kappa(x,y,t,\lambda)=\int_0^1\frac{\partial \mathcal R}{\partial\zeta}(\theta\Psi^-_1(x,y,t,\lambda)+(1-\theta)\Psi^-_2(x,y,t,\lambda),\lambda)d\theta.
\end{array}
\] Then
\[
\begin{array}{c}
\Pi^+(x,y,t,\lambda)=\kappa(x,y,t,\lambda)\Pi^-(x,y,t,\lambda),\\
\Pi^\pm(x,y,t,\lambda)\in H^p.
\end{array}
\]Therefore, justifying the index zero condition by  \eqref{E:uniq},  $|\kappa(x,y,0,\lambda)|=|\frac{\partial \mathcal R}{\partial\zeta}(\Phi^-(x,y,\lambda),\lambda)|\ne 0$, and   Lemma \ref{L:compact-y}, one can derive $\Pi^\pm(x,y,t,\lambda)\equiv 0$ and verify the uniqueness.

To complete the proof, we need to consider cases of higher derivatives. By \eqref{E:Manakov-Santini-NRH-local}, we obtain the Riemann-Hilbert problems
\bea
&\Psi_x^+={{\partial \mathcal R}/{\partial\zeta}(\Psi^-,\lambda)}\Psi^-_x,&
\label{E:Manakov-Santini-d-x}\\
&\Psi_y^+={{\partial \mathcal R}/{\partial\zeta}(\Psi^-,\lambda)}\Psi^-_y,&
\label{E:Manakov-Santini-d-y}\\
&\Psi_t^+={{\partial R}/{\partial\zeta}(\Psi^-,\lambda)}\Psi^-_t.&
\label{E:Manakov-Santini-d-t}
\eea
Using the renormalization $\Psi(x,y,t,\lambda)=\omega(x,y,t,\lambda)+x-\lambda y-\lambda^2 t$, the Riemann-Hilbert problems \eqref{E:Manakov-Santini-d-x}-\eqref{E:Manakov-Santini-d-t} turn into
\bea
&\omega_x^+=\frac{\partial  R}{\partial\zeta}\omega^-_x+\frac{\partial  R}{\partial\zeta},&
\label{E:Manakov-Santini-d-x-linear}\\
&\omega_y^+=\frac{\partial  R}{\partial\zeta}\omega^-_y-\frac{\partial  R}{\partial\zeta}\lambda,&
\label{E:Manakov-Santini-d-y-linear}\\
&\omega_t^+=\frac{\partial  R}{\partial\zeta}\omega^-_t-\frac{\partial  R}{\partial\zeta}\lambda^2,&
\label{E:Manakov-Santini-d-t-linear}
\eea
where $\frac{\partial  R}{\partial\zeta}=\frac{\partial  R}{\partial\zeta}(\omega^-+x-\lambda y-\lambda^2 t,\lambda)$.
They are linear Riemann-Hilbert problem \eqref{E:Manakov-Santini-d-x} with non homogeneous terms
\begin{equation}\label{E:nonhomogeneous-linear}
\frac{\partial  R}{\partial\zeta},\ \frac{\partial  R}{\partial\zeta}\lambda,\ \frac{\partial  R}{\partial\zeta}\lambda^2\in L^p,\quad p>2.
\end{equation}
Therefore, by Lemma \ref{L:compact-y} (cf. \eqref{E:non-homo} and \eqref{E:solution-diff}), one can derive the unique solvability in $L^p$ of \eqref{E:Manakov-Santini-d-x-linear}-\eqref{E:Manakov-Santini-d-t-linear}. 
 
Furthermore, for higher derivatives, we first use the same method to prove
\begin{equation}\label{E:reg-1}
\begin{array}{c}
\partial_t^h\partial_y^k\partial_x^\mu[\Psi^\pm-( x-\lambda y-\lambda^2 t)]\in L^p\cap L^{2p},\\
p>2,\quad 0\le h+k+\mu\le 1,
\end{array}
\end{equation} for  $0\le t\le \epsilon_2$, where $\epsilon_2$ is defined by \eqref{E:time}. On the other hand, taking derivatives of  \eqref{E:Manakov-Santini-d-x-linear}-\eqref{E:Manakov-Santini-d-t-linear}, one obtains
\begin{equation}\label{E:higher-derivatives}
\begin{split}
&\omega_{xx}^+=\frac{\partial  R}{\partial\zeta}\omega^-_{xx}+\frac{\partial^2  R}{\partial\zeta^2}\left(\ (\omega_x^-)^2+2\omega^-_x+1\right),\\
&\omega_{yt}^+=\frac{\partial  R}{\partial\zeta}\omega^-_{yt}+\frac{\partial^2  R}{\partial\zeta^2}\left(\omega_y^-\omega^-_t-\lambda^2\omega^-_y-\lambda\omega^-_t-\lambda^3\right),  \\
&\omega_{xy}^+=\frac{\partial  R}{\partial\zeta}\omega^-_{xy}+\frac{\partial^2  R}{\partial\zeta^2}\left(\omega_x^-\omega^-_y-\lambda\omega^-_x+\omega^-_y-\lambda\right),  \\
&\omega_{xt}^+=  \frac{\partial  R}{\partial\zeta}\omega^-_{xt}+\frac{\partial^2  R}{\partial\zeta^2}\left(\omega_x^-\omega^-_t-\lambda^2\omega^-_x+\omega^-_t-\lambda^2\right), \\
&\omega_{yy}^+=  \frac{\partial  R}{\partial\zeta}\omega^-_{yy}+\frac{\partial^2  R}{\partial\zeta^2}\left(\ (\omega_y^-)^2-2\lambda\omega^-_y+\lambda^2\right).
\end{split}
\end{equation} Theorem \ref{P:NRH-problem} and  \eqref{E:reg-1} imply all non homogeneous terms in the right hand sides of \eqref{E:higher-derivatives} are $L^p$ functions. Therefore, the estimates \eqref{E:reg} can be derived as above.

\end{proof}

\begin{remark} If the condition \eqref{E:reg} is dropped, the uniqueness is no longer true. A counterexample is given by $\Phi^\pm(x-\lambda^2 t, y,\lambda)$ which satisfies \eqref{E:Manakov-Santini-NRH-local} but does not satisfy \eqref{E:reg}. 
\end{remark}

\begin{remark}\label{R:newton} The Newtonian iteration in this section can be elucidated to prove the set 
\[
\begin{split}
S=\ &\{T\in\RR| \ \textit{the nonlinear Riemann-Hilbert problem}\\
&\ \Psi^+(x,y,t,\lambda)= \mathcal R(\Psi^-(x,y,t,\lambda),\lambda),\ \  \lambda\in\RR,\\
&\ \partial_{\overline\lambda}\Psi(x,y,t,\lambda)=0,\ \ \lambda\in\CC^\pm,\\
&\ \Psi(x,y,0,\lambda)= \Phi(x,y,\lambda)\\
&\ \textit{is solved for $0\le t\le T$}.\}
\end{split}
\]is open. On the other hand, the dispersion relation $x-\lambda y-\lambda^2 t$  causes the estimates, in particular,  
\[|R(\omega^\pm(x,y,t,\lambda)+x-\lambda y-\lambda^2 t,\lambda )|_{H^p(\RR,d\lambda)}\] growing  inevitably unbounded  as  $t\gg 1$. That is,  when  $t\gg 1$, one can not exclude the possibilities 
\begin{itemize}
\item the index zero condition on $1+\partial_\zeta R$ breaks;
\item the deformation property of $R(\Phi^-(x,y,\lambda),\lambda)$ fails,
\end{itemize}
since $\hbar<\infty$.  
 Therefore, $S$ is not closed in general. So only  a local solvability for the nonlinear Riemann-Hilbert problem \eqref{E:Manakov-Santini-NRH-local} is achieved and it is not practical to recover the global small data solution obtained in \cite{GSW}  via the Newtonian iteration scheme.

\end{remark}

\section{The inverse problem II: the Lax equation and the Cauchy problem}\label{S:lax-cauchy}

\begin{theorem}\label{T:eigenfunction-inv-0-S} {\bf (Local solvability of the Lax pair)} 
Suppose $v_0\in\mathfrak S$ is compactly supported in $y$. Let  $\Psi^\pm(x,y,t,\lambda)$ and $\epsilon_2=\epsilon_2(y)$ be the solution of the nonlinear Riemann-Hilbert problem (\ref{E:Manakov-Santini-NRH-local}) obtained in Theorem \ref{T:uniqueness} (replacing $v$ by $v_0$). Define
\begin{gather}
v(x,y,t)=\frac 1{2\pi i}\int_{\mathbb R}R(\Psi^-(x,y,t,\lambda'),\lambda')d\lambda',\label{E:potential-inv-S}\\
\Psi(x,y,t,\lambda)=x-\lambda y-\lambda^2 t+\frac 1{2\pi i}\int_{\mathbb R}\frac {R(\Psi^-(x,y,t,\lambda'),\lambda')}{\lambda'-\lambda}d\lambda', \label{E:eigenfunction-inv-S}\\
 \forall x,\, y\in\RR,\ \ \lambda\in\mathbb C ^\pm,\ \ 0\le t\le \epsilon_2.\nonumber
\end{gather}
Then for $ 0\le \mu+k\le 1$,  
\begin{gather}
v(x,y,t)=\overline {v(x,y,t)}, \label{E:potential-regularity-S}\\
v(x,y,0)=v_0(x,y),\label{E:initial-v-S}\\
\partial_y^k\partial_x^\mu v(x,y,t)\in C(\RR\times\RR\times  [0,\epsilon_2])\cap L^\infty(\RR\times\RR\times  [0,\epsilon_2]);\label{E:eigenfunction-inv-1-S}
\end{gather} and for  $0\le \mu+k+h \le 2$, $h<2$, and $\lambda\in\CC^\pm$,
\begin{equation}\label{E:eigen-condition-x-S}
 \partial_t^h\partial_y^k\partial_x^\mu[\Psi-(x-\lambda y-\lambda^2 t)]
   \in L^\infty(\RR\times\RR\times [0,\epsilon_2]).
   \end{equation}
Moreover,  for   $\lambda\in\CC^\pm$, the Lax pair
\begin{gather}
L\Psi=\partial_y\Psi+\left(\lambda+v_x\right)\partial_x \Psi=0,\label{E:pavlov-inv-x-S}\\
M\Psi=\partial_t\Psi+\left(\lambda^2+\lambda v_x-v_y\right)\partial_x\Psi=0, \label{E:pavlov-inv-y-S}\\
\textit{ for $\forall x$, $y\in\RR$, $0\le t\le \epsilon_2$}\nonumber
\end{gather}
exist uniquely. 
\end{theorem}
\begin{proof} Theorem  \ref{T:uniqueness}, \eqref{E:reg},   (\ref{E:potential-inv-S}), \eqref{E:eigenfunction-inv-S}, Sobolev's theorem, and  H$\ddot{\textrm o}$lder  inequality imply   
 (\ref{E:eigenfunction-inv-1-S}) and   (\ref{E:eigen-condition-x-S}). 

Applying $L$ and $M$ to both sides of \eqref{E:Manakov-Santini-NRH-local}, and using \eqref{E:reg}, \eqref{E:eigenfunction-inv-1-S}, $p>2$,  we obtain
\beq\label{E:lax-liouville-1}
\begin{split}
L\Psi^+=\frac{\partial{\mathcal R}}{\partial \zeta}L\Psi^-,&\quad \textit{$L\Psi^\pm\in L^p$,}\\
M\Psi^+=\frac{\partial \mathcal R}{\partial\zeta}M\Psi^-,&\quad \textit{$M\Psi^\pm\in L^p$.}
\end{split}
\eeq
Applying  Lemma \ref{L:compact-y} (cf. \eqref{E:non-homo} and \eqref{E:solution-diff}), we conclude $L\Psi^\pm=L\Psi=0$ and $M\Psi^\pm=M\Psi=0$.  Hence (\ref{E:pavlov-inv-x-S}) and \eqref{E:pavlov-inv-y-S} are justified.

From \eqref{E:psi-reality}, one can prove $\Psi(x,y,t,\lambda)=\overline{\Psi(x,y,t,\bar\lambda)}$. Together with \eqref{E:pavlov-inv-x-S}, we obtain the reality condition (\ref{E:potential-regularity-S}).

\end{proof}

\begin{theorem}\label{T:inverse-problem-P-1-S} {\bf (Local solvability of the Cauchy problem)} 
Suppose $v_0\in\mathfrak S$ with a compact support. Then there exists $\epsilon_2(y)>0$ such that the Cauchy problem of the Pavlov equation
\begin{equation}\label{E:cauchy-pavlov-S}
\begin{split}
&v_{xt}+v_{yy}=v_yv_{xx}-v_xv_{xy},\quad \forall x,\,y\in\RR,\ 0< t\le \epsilon_2(y),\\
&v(x,y,0)=v_0(x,y) 
\end{split}
\end{equation}admits a unique 
real solution. Moreover, 
\[
\begin{array}{c}
\partial_y^k\partial_x^\mu v\in C(\RR\times\RR\times  [0,\epsilon_2])\cap L^\infty(\RR\times\RR\times  [0,\epsilon_2]),\\
 0\le \mu+k\le 1.
 \end{array}
\]
\end{theorem}

\begin{proof} Applying Theorem \ref{T:eigenfunction-inv-0-S} and  Lemma \ref{L:decay}, one can derive higher order (non uniform) decay of $R(\Psi(x,y,t,\lambda)$ 
 in $\lambda$ if $0<t$. Hence $v_{xt}$, $v_{yy}$, $v_{xy}$, and $v_{xx}$ exist. So we can compute the compatibility of the Lax pair (\ref{E:pavlov-inv-x-S}) and (\ref{E:pavlov-inv-y-S}) and obtain
\[
\left(v_{xt}+v_{yy}-v_yv_{xx}+v_xv_{xy}\right)\partial_x\Psi\equiv 0,\quad 0< t\le \epsilon_2(y).
\]

\end{proof}

\begin{remark} Unlike Theorem \ref{T:eigenfunction-inv-0-S}  where the Lax pair holds up to $t=0$, we cannot prove the Pavlov equation is valid at $t=0$.
\end{remark}

\begin{remark}
The property $\hbar<\infty$ in \eqref{E:holo-ext} is the obstruction to global solvability of the Lax pair and Cauchy problem with large initial data.
\end{remark}

\end{document}